\documentclass{article}

\usepackage{arxiv}

\usepackage[utf8]{inputenc} 
\usepackage[T1]{fontenc}    
\usepackage[numbers]{natbib}

\usepackage{hyperref}       
\usepackage{url}            
\usepackage{booktabs}       
\usepackage{amsfonts}       
\usepackage{nicefrac}       
\usepackage{microtype}      
\usepackage{graphicx}
\usepackage{doi}

\usepackage{amsmath}
\usepackage{amssymb}
\usepackage{latexsym}
\usepackage{enumerate}
\usepackage{colortbl}
\usepackage{amsthm}
\usepackage{bm}
\usepackage{xcolor}
\usepackage{tikz, pgfplots, standalone} 
\pgfplotsset{compat=1.18} 
\usepackage{subcaption}  

\newtheorem{thm}{Theorem}
\newtheorem{lem}{Lemma}
\newtheorem{prop}{Proposition}

\newtheorem{cor}{Corollary}

\newcommand{\argmax}{\mathop{\rm arg~max}\limits}

\newcommand{\zb}{\bm{z}}

\newcommand{\ub}{\bm{u}}
\newcommand{\rhob}{\bm{\rho}}
\newcommand{\phib}{\bm{\phi}}

\newcommand{\R}{\mathbb R}

\title{The trivariate wrapped Cauchy copula}

\author{ \href{https://orcid.org/0000-0003-2648-6769}{\includegraphics[scale=0.06]{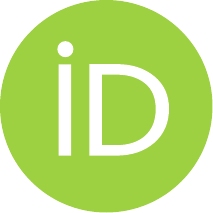}\hspace{1mm}Shogo Kato} \\
    Institute of Statistical Mathematics \\
    Japan \\
	\texttt{skato@ism.ac.jp} \\
	\And
	\href{https://orcid.org/0000-0002-2290-8437}{\includegraphics[scale=0.06]{orcid.pdf}\hspace{1mm}Christophe Ley} \\
	University of Luxembourg\\
    Luxembourg \\
	\texttt{christophe.ley@uni.lu} \\
    \And
	\href{https://orcid.org/0009-0006-1372-7032}{\includegraphics[scale=0.06]{orcid.pdf}\hspace{1mm}Sophia Loizidou} \\
	University of Luxembourg\\
    Luxembourg \\
	\texttt{sophia.loizidou@uni.lu} \\
    \And
	\href{https://orcid.org/0000-0003-0090-6235}{\includegraphics[scale=0.06]{orcid.pdf}\hspace{1mm}Kanti V. Mardia} \\
    University of Leeds \\
    United Kingdom \\
	\texttt{k.v.mardia@leeds.ac.uk} \\
}

\renewcommand{\shorttitle}{The trivariate wrapped Cauchy copula}

\hypersetup{
pdftitle={The trivariate wrapped Cauchy copula},
pdfauthor={S. Kato, C. Ley, S. Loizidou, K. V. Mardia},
pdfkeywords={Angular data, copula, directional statistics, flexible modeling, wrapped Cauchy distribution},
}

\begin{document}
\maketitle

\begin{abstract}
In this paper, we propose a new flexible distribution for data on the three-dimensional torus which we call a trivariate wrapped Cauchy copula. Our trivariate copula has several attractive properties. It has a simple form of density and desirable modality properties. Its parameters  allow for an adjustable degree of dependence between every pair of variables and these can be easily estimated. The conditional distributions of the model are well studied bivariate wrapped Cauchy distributions. Furthermore, the distribution can be easily simulated. Parameter estimation via maximum likelihood for the distribution is given and we highlight the simple implementation procedure to obtain these estimates. We illustrate our trivariate wrapped Cauchy copula on data from protein bioinformatics of conformational angles.
\end{abstract}


\section{Introduction}\label{sec:Intro}

Angular data occurs frequently in domains such as environmental sciences (e.g., wind directions, wave directions), bioinformatics (dihedral angles in protein backbone structures), zoology (animal movement studies), medicine (circadian body clock, secretion times of hormones), or political/social sciences (times of crimes during the day), see for example \cite{LV19}. However, dealing with data that are angles requires special care in order to take into account their topology (domain $[0,2\pi)$ where the end points coincide). 
Classical statistical concepts from the real line no longer hold for such data \cite{jammalamadaka_topics_2001, mardia2000}; in particular, the building blocks of statistical modeling and inference, probability distributions, need to be properly defined. For data involving a single angle, called circular data, a large body of literature proposing circular distributions exists (see for example \cite{mardia_directional_2018} and \cite{pewsey_recent_2021}), among which the popular von Mises, wrapped Cauchy and cardioid models. 
Though less in number, there also exist several interesting distributions for data consisting respectively of two angles, called toroidal data, and an angle and a linear part, called cylindrical data. 
Popular examples for toroidal distributions are the bivariate von Mises \cite{mardia1975}, the Sine \cite{singh2002} and Cosine models \cite{mardia2007} and the bivariate wrapped Cauchy \cite{KP15}, while examples for cylindrical models are {the Johnson--Wehrly \cite{johnson_angular-linear_1978}, the Mardia--Sutton \cite{mardia_cylindrical_1978}, the Kato--Shimizu \cite{kato_dependent_2008} and the Abe--Ley \cite{abe_tractable_2017} distributions. }

The situation is somewhat different for the case of trivariate distributions and as far as we know  there is no model which has all of  the properties our  model has. Namely, explicit normalizing constant, easy way to simulate, and  interpretable parameters.   However, many  datasets {exist} which require such a model. An important example stems from structural bioinformatics.  Predicting the three-dimensional folding structure of a protein from its known one-dimensional amino acid structure is among the most important yet hardest scientific challenges, with impacts in drug development, vaccine design, disease mechanism understanding, human cell injection, and enzyme engineering, to cite but these. So far most statistical advances, in particular on flexible and tractable probability distributions, have considered the two dihedral angles $\phi$ and $\psi$, and considered the torsion angle of the side chain $\omega$ to be fixed at either 0 or $\pi$ (which are the only two realistic values for this angle). One exception is  \cite{BMTFKH08} where  $\omega$ is taken as binary random variable with value either 0 or $\pi$ sampled using the hidden Markov model  though the distribution is discrete not continuous. In practice this angle $\omega$ is often measured with some noise, hence a model for $\phi, \psi$ and $\omega$ is required.


In order to fill this important gap in the literature, we propose in this paper the trivariate wrapped Cauchy copula. Note that \cite{circulas} coined the term \emph{circulas} for copulas on the torus.
Circulas differ fundamentally from rescaled copulas in that they also require the periodicity of the densities, see \cite{circulas}, Section 1. Our trivariate wrapped Cauchy copula has the following further benefits: (i) simple form of density, (ii) adjustable degree of dependence between the variables, (iii) interpretable and well-estimatable parameters, (iv) well-known conditional distributions, and (v) a simple data generating mechanism. Thus we are proposing a new model that satisfies the requirements as laid out in \cite{LBC21}.

The paper is organized as follows. In Section~\ref{sec: construction} we review some well-known  bivariate distributions which are relevant for our construction before proposing the trivariate wrapped Cauchy copula in Section~\ref{sec: definition}.
In Section~\ref{sec:Props}, several properties of the proposed distribution are given, including conditional distributions, random variate generation, and modality. Parameter estimation is considered in Section~\ref{sec: parameter estimation} and a real-data analysis provided in Section~\ref{sec:real_data}. Section~\ref{sec:conclusion} concludes the paper with a discussion. We provide  additional information in the Supplementary Material, namely an essential lemma about parameter space, trigonometric moments, correlation coefficients, further properties, method of moment estimation, Monte Carlo simulation results, a multivariate extension and a generalization of our new distribution.


\section{Bivariate circular probability distributions and copulas} \label{sec: construction}

Before we introduce our trivariate wrapped Cauchy copula, we will first review the most popular bivariate circular distributions.

A popular class of toroidal (circular-circular) distributions that allows specifying the marginal distributions has been put forward by Wehrly and Johnson (1980) \cite{WJ80}. 
Their general families have the probability density functions 
\begin{equation}
	f(\theta_1,\theta_2)= 2 \pi g [2 \pi \{ F_1(\theta_1) -q
	F_2(\theta_2) \}] f_1 (\theta_1) f_2 (\theta_2), \label{eq:wj1}
\end{equation}
where $q$ equals 1 or -1 (leading to two distinct families), $0 \leq
\theta_1,\theta_2 <  2 \pi,$ $f_1$ and $f_2$ are specified densities on the circle $[0,2\pi)$, $F_1$ and
$F_2$ are their distribution functions defined with respect to
fixed, arbitrary, origins, and $g$ is also a specified density on the circle. Both families (\ref{eq:wj1}) have the nice property that their marginal densities are given by $f_1$ and $f_2$.
More precisely, let a bivariate circular random vector $(\Theta_1, \Theta_2)$ have the distribution (\ref{eq:wj1}).
Then the marginal densities of $\Theta_1$ and $\Theta_2$ are given by $f_1$ and $f_2$, respectively. Note that expression~\eqref{eq:wj1} with $\theta_2$ replaced by $x\in\R$ (and accordingly $f_2$ and $F_2$ are densities and distribution functions on $\R$) had been proposed by Johnson and Wehrly (1978) \cite{johnson_angular-linear_1978} to build general circular-linear distributions with specified marginals.

The  families (\ref{eq:wj1})  can be readily transformed into copulas for bivariate circular data, assuming $(U_1 ,U_2)' = (2 \pi F_1(\Theta_1), 2 \pi F_2(\Theta_2))'$.
Throughout this paper, unlike the case of ordinary copulas, we define the ranges of the copula variables to be $[0, 2\pi)$ or equivalent ranges modulo $2\pi$, rather than $[0, 1]$.
Then the density of $(U_1,U_2)'$ is given by
\begin{equation}
	c(u_1,u_2) = \frac{1}{2 \pi} g (u_1 - q u_2), \quad 0 \leq u_1,u_2 < 2\pi. \label{eq:jones_etal}
\end{equation}
 Simple integration shows that $c(u_1,u_2)$ integrates to 1 and that each marginal is uniform on the circle (see \cite{circulas} for further properties of the distribution (\ref{eq:jones_etal})).
Therefore the distribution (\ref{eq:jones_etal}) can be viewed as an equivalent of a copula for bivariate circular data.
From Sklar's theorem (see page 18 in \cite{nelsen_introduction_2006}), the distribution with density (\ref{eq:jones_etal}) can be transformed into a distribution with prespecified marginal distributions, as is the case for the distributions of Wehrly and Johnson (\ref{eq:wj1}).

In practice, it is necessary to take specific densities for $f_1$, $f_2$ and $g$ in the families of Wehrly and Johnson, or equivalently, the choice of $g$ and marginal distributions is important for the copula-based version (\ref{eq:jones_etal}).
Various proposals have been put forward in the literature, many of which are based on the use of the von Mises distribution.
We refer the reader to \cite{ley2017}, Section 2.4, for a review of these models. Further, \cite{K09} and \cite{KP15} adopted the wrapped Cauchy density as the function $g$ in (\ref{eq:jones_etal}).
This leads to density (\ref{eq:jones_etal}) being of the form
\begin{equation}
c(u_1,u_2) = \frac{1}{4 \pi^2} \frac{|1-\rho^2|}{1+\rho^2 - 2 \rho \cos (u_1-u_2 - \mu)}, \quad 0 \leq u_1,u_2 < 2\pi,  \label{eq:wc_wj}
\end{equation}
where $\mu \in [0,2\pi )$ is the location parameter and $\rho \in \mathbb{R} \setminus \{\pm 1\}$ the dependence parameter between the two variables.
{The distribution (\ref{eq:wc_wj}) is called a bivariate wrapped Cauchy copula which we label as BWC($\mu,\rho$).}
In fact, \cite{K09} showed that this distribution is derived from a problem in Brownian motion and possesses various tractable properties.
Moreover, \cite{KP15} transformed this model via the M\"obius transformation and showed that the transformed distribution (known as Kato--Pewsey model) has the wrapped Cauchy marginal and conditional distributions.

\section{Our proposal: the trivariate wrapped Cauchy copula} \label{sec: definition}

Having described some of the attractive properties of the bivariate wrapped Cauchy distribution of \cite{KP15} resulting from the copula~\eqref{eq:wc_wj}, we extend~\eqref{eq:wc_wj} to a density on the three-dimensional torus in the following theorem, with anticipation that it will inherit these properties.  

\begin{thm} \label{thm:density}
For $\ub = \left( u_1, u_2, u_3 \right)'$ and $\rhob = \left( \rho_{12}, \rho_{13}, \rho_{23} \right)'$, let
\begin{align}
	\hspace{-0.1cm} t(\ub; \rhob) & =  c_2 \Bigl[ c_1 + 2 \left\{ \rho_{12} \cos (u_1 - u_2) + \rho_{13} \cos (u_1 - u_3) + \rho_{23} \cos (u_2 - u_3) \right\} \Bigr]^{-1},    \nonumber \\
	& \hspace{8cm} 0 \leq u_1,u_2, u_3 < 2\pi,  \label{eq:tri_density} 
\end{align}
where $\rho_{12},\rho_{13},\rho_{23} \in \mathbb{R} \setminus \{0\} $, $\rho_{12}\rho_{13} \rho_{23} >0$,
\begin{equation}
c_1 = \frac{\rho_{12} \rho_{13}}{ \rho_{23}} + \frac{\rho_{12} \rho_{23}}{\rho_{13}} + \frac{ \rho_{13} \rho_{23} }{\rho_{12}}, \label{eq:c1}
\end{equation}
and
\begin{equation} \label{eq:c2}
c_2 = \frac{1}{(2\pi)^3} \left\{ \left( \frac{\rho_{12} \rho_{13}}{ \rho_{23}} \right)^2 + \left( \frac{\rho_{12} \rho_{23}}{\rho_{13}} \right)^2 + \left(\frac{\rho_{13} \rho_{23}}{\rho_{12}} \right)^2 - 2 \rho_{12}^2 - 2 \rho_{13}^2 - 2 \rho_{23}^2 \right\}^{1/2}.
\end{equation}
Suppose that there exists one of the permutations of $(1,2,3)$, $(i,j,k)$, such that 
\begin{equation}\label{conditions}
    |\rho_{j k}| < |\rho_{ij} \rho_{i k}| / ( |\rho_{ij}| + |\rho_{i k}|),\end{equation} where $\rho_{ji} = \rho_{ij}$ for $1 \leq i < j \leq 3$.
Then the function (\ref{eq:tri_density}) is a probability density function on the three-dimensional torus $[0,2 \pi)^3$. Its parameters $\rhob$ are not identifiable but without loss of generality we take our identifiability constraint given by \begin{equation}\label{identcond}
\rho_{12}\rho_{13}\rho_{23}=1.
\end{equation}
\end{thm}

\begin{proof}
In order to prove the theorem, it suffices to see that the function (\ref{eq:tri_density}) satisfies: 
\begin{itemize}
\item[(i)] $t(\ub; \rhob) \geq 0$ for any $\ub = (u_1,u_2,u_3)'$ ; 
\item[(ii)] $\int_{[0,2\pi)^3} t(\ub; \rhob) du_1 du_2$ $du_3 =1$.
\end{itemize}
\vspace{0.2cm} 

\textit{Proof of (i):}
Without loss of generality, assume for \eqref{conditions} that $|\rho_{12}| < |\rho_{13} \rho_{23}|/(|\rho_{13}| + |\rho_{23}|)$.
We note that
\begin{align*}
\lefteqn{ c_1 + 2 \left\{ \rho_{12} \cos (u_1 - u_2) + \rho_{13} \cos (u_1 - u_3) + \rho_{23} \cos (u_2 - u_3) \right\} } \hspace{1cm} \\
 & = \left\| \phi_1
\begin{pmatrix}
\cos u_1 \\ \sin u_1
\end{pmatrix} 
+ \phi_2
\begin{pmatrix}
\cos u_2 \\ \sin u_2
\end{pmatrix}
+ \phi_3 
\begin{pmatrix}
\cos u_3 \\ \sin u_3
\end{pmatrix}
\right\|^2,
\end{align*}
where 
\begin{equation}\label{eqrelat}
{ \phi_i = \mbox{sgn} (\rho_{jk}) \left| \frac{\rho_{ij} \rho_{ik}}{\rho_{jk}} \right|^{1/2}, \quad i,j,k=1,2,3 , \quad j,k \neq i, \quad j <k. }
\end{equation}
Then
\begin{align}
\lefteqn{ \left\| \phi_1
\begin{pmatrix}
\cos u_1 \\ \sin u_1
\end{pmatrix} 
+ \phi_2
\begin{pmatrix}
\cos u_2 \\ \sin u_2
\end{pmatrix}
+ \phi_3 
\begin{pmatrix}
\cos u_3 \\ \sin u_3
\end{pmatrix}
\right\| } \hspace{1cm} \nonumber \\
 & \geq \left\| \phi_3 
\begin{pmatrix}
\cos u_3 \\ \sin u_3
\end{pmatrix}
\right\| - \left\| \phi_1
\begin{pmatrix}
\cos u_1 \\ \sin u_1
\end{pmatrix} 
+ \phi_2
\begin{pmatrix}
\cos u_2 \\ \sin u_2
\end{pmatrix}
\right\|  \nonumber \\
 & \geq \sqrt{ \frac{\rho_{13} \rho_{23}}{\rho_{12}} } - \left( \sqrt{\frac{\rho_{12} \rho_{13}}{\rho_{23}}} + \sqrt{ \frac{\rho_{12} \rho_{23}}{\rho_{13}} }  \right) > 0. \label{eq:positive}
\end{align}
Therefore $c_1 + 2 \left\{ \rho_{12} \cos (u_1 - u_2) + \rho_{13} \cos (u_1 - u_3) + \rho_{23} \cos (u_2 - u_3) \right\} > 0$ for any $(u_1,u_2,u_3)'$.

Next we show that the expression between curly brackets in $c_2$ given in (\ref{eq:c2}) is positive.
It is easy to see that
\begin{align}
 \lefteqn{\rho_{12}^2\rho_{13}^2\rho_{23}^2 \left(\left( \frac{\rho_{12} \rho_{13}}{ \rho_{23}} \right)^2 + \left( \frac{\rho_{12} \rho_{23}}{\rho_{13}} \right)^2 + \left(\frac{\rho_{13} \rho_{23}}{\rho_{12}} \right)^2 - 2 \rho_{12}^2 - 2 \rho_{13}^2 - 2 \rho_{23}^2  \right)}  \hspace{1cm} \nonumber \\
 & = (\rho_{12}^2 - \rho_{23}^2 )^2 \left\{ \rho_{13}^2 - \frac{\rho_{12}^2 \rho_{23}^2 (\rho_{12}^2 + \rho_{23}^2)}{(\rho_{12}^2 - \rho_{23}^2)^2} \right\}^2 - \frac{4 \rho_{12}^6 \rho_{23}^6}{(\rho_{12}^2 - \rho_{23}^2)^2}. \label{eq:radicand}
\end{align}
It follows from the assumption $|\rho_{12}| < |\rho_{13} \rho_{23}|
/(|\rho_{13}| + |\rho_{23}|)$ that $|\rho_{13}| > |\rho_{12} \rho_{23}| / (|\rho_{23}|-|\rho_{12}|)$.
Also,
$
\rho_{13}^2 > 
\{ |\rho_{12} \rho_{23}| / (|\rho_{23}|-|\rho_{12}| ) \}^2.
$
Then the radicand (\ref{eq:radicand}) can be evaluated as
\begin{align*}
\lefteqn{ (\rho_{12}^2 - \rho_{23}^2 )^2 \left\{ \rho_{13}^2 - \frac{\rho_{12}^2 \rho_{23}^2 (\rho_{12}^2 + \rho_{23}^2)}{(\rho_{12}^2 - \rho_{23}^2)^2} \right\}^2 - \frac{4 \rho_{12}^6 \rho_{23}^6}{(\rho_{12}^2 - \rho_{23}^2)^2} } \hspace{1cm} \\
& > (\rho_{12}^2 - \rho_{23}^2 )^2 \left\{   \left( \frac{ |\rho_{12} \rho_{23}| }{|\rho_{23}|-|\rho_{12}|} \right)^2 - \frac{\rho_{12}^2 \rho_{23}^2 (\rho_{12}^2 + \rho_{23}^2)}{(\rho_{12}^2 - \rho_{23}^2)^2} \right\}^2 - \frac{4 \rho_{12}^6 \rho_{23}^6}{(\rho_{12}^2 - \rho_{23}^2)^2} \\
& =0.
\end{align*}
Hence $c_2 >0$.
Thus the function (\ref{eq:tri_density}) satisfies the condition (i) for $|\rho_{12}| < |\rho_{13} \rho_{23}|/(|\rho_{13}| + |\rho_{23}|)$. 
Due to the symmetry of the function (\ref{eq:tri_density}), it immediately follows that the condition (i) holds for the other two cases $|\rho_{13}|<|\rho_{12} \rho_{23}|/(|\rho_{12}| + |\rho_{23}|)$ and $|\rho_{23}|<|\rho_{12} \rho_{13}|/(|\rho_{12}| + |\rho_{13}|)$. \vspace{0.3cm}

\textit{Proof of (ii):}
Using the equation (3.613.2.6) of \cite{GR2007}, it follows that
\begin{align}
			\lefteqn{ \int_0^{2\pi}  t(\ub; \rhob) d u_3 } \nonumber \hspace{0.1cm} \\
   & =  \frac{c_2}{c_1 + 2 \rho_{12} \cos (u_1-u_2)} \int_0^{2\pi} \left[ 1 + \frac{2 \{ \rho_{13} \cos (u_1 - u_3) + \rho_{23} \cos (u_2 - u_3) \}}{ c_1 + 2 \rho_{12} \cos (u_1-u_2) } \right]^{-1} d u_3 \nonumber \\
			& =  \frac{c_2}{c_1 + 2 \rho_{12} \cos (u_1-u_2)} \int_0^{2\pi} \frac{1}{1+b' \cos (u_3 - a') } d u_3 \nonumber \\
			& = \frac{c_2}{c_1 + 2 \rho_{12} \cos (u_1-u_2)} \cdot \frac{2\pi}{|1-b'^2|^{1/2}} \label{eq:marginal2_proof}, 
		\end{align}
where $b'=2 \{ \rho_{13}^2 + \rho_{23}^2 + 2 \rho_{13} \rho_{23} \cos (u_1-u_2)\}^{1/2} / \{c_1 + 2 \rho_{12} \cos (u_1 - u_2)\}$ and $a'$ satisfies $\tan \alpha' = (\rho_{13} \sin u_1 + \rho_{23} \sin u_2)/ ( \rho_{13} \cos u_1 + \rho_{23} \cos u_2 )$.
The second equality follows from the formula $\alpha \cos u_3 + \beta \sin u_3 = \sqrt{\alpha^2 + \beta^2 } \cos (u_3 - \gamma)$, where $\gamma$ satisfies $\tan\gamma = \beta/\alpha$.
In order to see $|b'|<1$ in the last equality, it suffices to see that
\begin{align*}
\lefteqn{ c_1 + 2\rho_{12} \cos (u_1-u_2) - 2 \{ \rho_{13}^2 + \rho_{23}^2 + 2 \rho_{13} \rho_{23} \cos (u_1-u_2)\}^{1/2} } \hspace{1.5cm} \\
& =
\left|  \left\| \phi_1
\begin{pmatrix}
\cos u_1 \\ \sin u_1
\end{pmatrix} 
+ \phi_2
\begin{pmatrix}
\cos u_2 \\ \sin u_2
\end{pmatrix} \right\|
-  \left\| \phi_3 
\begin{pmatrix}
\cos u_3 \\ \sin u_3
\end{pmatrix}
\right\| \right|^2 > 0,
\end{align*}
holds for any $u_1,u_2,u_3 \in [0,2 \pi)$.
If $|\rho_{12}| < |\rho_{13} \rho_{23}|/(|\rho_{13}|+|\rho_{23}|)$, this inequality is already seen in (\ref{eq:positive}).
If $|\rho_{13}| < |\rho_{12} \rho_{23}|/(|\rho_{12}| + |\rho_{23}|)$ or $|\rho_{23}| < |\rho_{12} \rho_{13}|/(|\rho_{12}| + |\rho_{13}|)$, we have
\begin{align*}
\lefteqn{ \left\| \phi_1
\begin{pmatrix}
\cos u_1 \\ \sin u_1
\end{pmatrix} 
+ \phi_2
\begin{pmatrix}
\cos u_2 \\ \sin u_2
\end{pmatrix}
\right\| - \left\| \phi_3 
\begin{pmatrix}
\cos u_3 \\ \sin u_3
\end{pmatrix}
\right\| } \hspace{2cm} \\
& \  \geq \left| \sqrt{\frac{\rho_{12} \rho_{13}}{\rho_{23}}} - \sqrt{ \frac{\rho_{12} \rho_{23}}{\rho_{13}} }  \right| - \sqrt{ \frac{\rho_{13} \rho_{23}}{\rho_{12}} } >0.
\end{align*}
The last inequality follows from 
\begin{align*}
\lefteqn{ \left| \sqrt{\frac{\rho_{12} \rho_{13}}{\rho_{23}}} - \sqrt{ \frac{\rho_{12} \rho_{23}}{\rho_{13}} }  \right|^2 - \frac{\rho_{13} \rho_{23}}{\rho_{12}} } \hspace{1.5cm} \\
 & = 
\frac{1}{\rho_{12} \rho_{13} \rho_{23}} \left[ \rho_{12}^2 \{ \rho_{13} - \rho_{23} \}^2 - \rho_{13}^2 \rho_{23}^2 \right] \\
 & \geq 
\frac{1}{\rho_{12} \rho_{13} \rho_{23}} \left\{ \rho_{12}^2 (|\rho_{13}| - |\rho_{23}| )^2 - \rho_{13}^2 \rho_{23}^2 \right\} \\
 & > \frac{1}{\rho_{12} \rho_{13} \rho_{23}} \left\{ \left( \frac{\rho_{13} \rho_{23}}{|\rho_{13}|-|\rho_{23}|} \right)^2 (|\rho_{13}| - |\rho_{23}|)^2 - \rho_{13}^2 \rho_{23}^2 \right\} = 0.
\end{align*}
Thus we have
$
c_1 + 2\rho_{12} \cos (u_1-u_2) - 2 \{ \rho_{13}^2 + \rho_{23}^2 + 2 \rho_{13} \rho_{23} \cos (u_1-u_2)\}^{1/2} > 0,
$
which implies $|b'|<1$.

(\ref{eq:marginal2_proof}) can be simplified as
\begin{align}
\int_0^{2\pi}  t(\ub ; \rhob) d u_3  & = \frac{2 \pi \cdot c_2}{|\rho_{12} \rho_{23} / \rho_{13} +  \rho_{12} \rho_{13} / \rho_{23} -\rho_{13} \rho_{23}/ \rho_{12} + 2 \rho_{12} \cos (u_1-u_2)|} \nonumber \\
& \equiv t_2(u_1, u_2; \phi_{12}), \label{eq:marginal_proof}
\end{align}
where $\phi_{12}$ is defined in \eqref{thephi}.
Here we show that the denominator of (\ref{eq:marginal_proof}) is positive for any $(u_1,u_2)'$.
For convenience, write $c_3= \rho_{12} \rho_{23} / \rho_{13} +  \rho_{12} \rho_{13} / \rho_{23} - \rho_{13} \rho_{23} / \rho_{12}$.
Let $\rho_{12}>0$.
For the case $|\rho_{12}| < |\rho_{13} \rho_{23}| / (|\rho_{13}| + |\rho_{23}|) $, it can be seen that
\begin{align*}
 c_3 + 2 \rho_{12} \cos (u_1-u_2) & \leq c_3 + 2 \rho_{12} = \frac{\rho_{12}^2 (\rho_{13}+\rho_{23})^2 - \rho_{13}^2 \rho_{23}^2}{\rho_{12} \rho_{13} \rho_{23}} \\
 & < \frac{1}{\rho_{12} \rho_{13} \rho_{23}} \left[ \left( \frac{|\rho_{13}| |\rho_{23}|}{|\rho_{13}| + |\rho_{23}|} \right)^2 (\rho_{13} + \rho_{23})^2 - \rho_{13}^2 \rho_{23}^2 \right] \leq 0.
\end{align*}
The case $|\rho_{13}| < |\rho_{12} \rho_{23}| / (|\rho_{12}| + |\rho_{23}|) $ or $|\rho_{23}| < |\rho_{12} \rho_{13}| / (|\rho_{12}| + |\rho_{13}|) $ can be proved in a similar manner.

Similarly, when $\rho_{12}<0$, it can be seen that the denominator of (\ref{eq:marginal_proof}) remains positive for any $(u_1, u_2)'$.

Note that the discussion above implies $|2\rho_{12} / c_3| <1$.
Then it follows from the equation (3.613.2.6) of \cite{GR2007} that, for the parameters satisfying $c_3+2\rho_{12} \cos (u_1-u_2) > 0$,
\begin{align}
\begin{split} \label{eq:uniform}
\int_{0}^{2\pi}  t_2(u_1, u_2; \phi_{12}) du_2 
& = \int_0^{2 \pi} \frac{2 \pi c_2}{ c_3 + 2 \rho_{12} \cos (u_1-u_2) } du_2 \\
& = \int_0^{2 \pi} \frac{2\pi c_2}{c_3 \{ 1 + 2 \, ( \rho_{12}/c_3) \cos (u_1-u_2) \} } du_2 \\
& = \frac{(2\pi)^2 c_2}{c_3 |1 - (2 \rho_{12}/c_3)^2 |^{1/2}} = \frac{(2\pi)^2 c_2}{|c_3^2 - 4 \rho_{12}^2|^{1/2}} = \frac{1}{2\pi} . 
\end{split}
\end{align}
Similarly, we can also show $\int_0^{2\pi} t_2(u_1, u_2; \phi_{12}) du_2 = 1/(2\pi)$ for the parameters which satisfy $c_3+2\rho_{12} \cos (u_1-u_2) < 0$.
Finally, 
$$
\int_{[0,2\pi)^3} t(\ub; \rhob) du_1 du_2 du_3 = \int_0^{2\pi} \frac{1}{2\pi} du_1 =1.
$$
Thus the proposed density (\ref{eq:tri_density}) satisfies the condition (ii) as required.

Since the function (\ref{eq:tri_density}) satisfies both conditions (i) and (ii), this function is a probability density function on $[0,2\pi)^3$.

We can show that the model is not identifiable by noting that  $t(\ub; \rhob)=t(\ub; c\rhob)$ for any constant $c\in\R^+$ and the condition \eqref{identcond} takes that into account.
\end{proof}

{We will refer to this distribution, which has the density of  $(U_1,U_2,U_3)'$ given by (\ref{eq:tri_density}), as the     trivariate wrapped Cauchy copula (TWCC or TWC copula).
The distribution (\ref{eq:tri_density}) is also denoted by TWCC$(\rhob)$ in order to specify its parameters.}
In {TWCC}, different functions of parameters $\rho_{12}, \rho_{13}$ and $\rho_{23}$ control  the dependence between the $U_i$'s and the location of the modes, {see Theorem~\ref{thm:modes} in Section~\ref{sec:Props} and Theorem~\ref{thm:correlation} from the Supplementary Material.} It is also straightforward to incorporate both positive and negative associations by replacing $u_i$ with $q_i u_i$ {$(q_i=1,-1)$} and it is possible to extend the distribution to include location parameters by replacing $u_i-u_j$ with $u_i - u_j-\mu_{ij}$ in (\ref{eq:tri_density}) $(0 \leq \mu_{ij} < 2\pi)$.
{Although the parametrization used in the expression (\ref{eq:tri_density}) for the density is natural in the sense that each parameter $\rho_{ij}$ is multiplied by a function of $u_i - u_j$, this expression does not allow for {the independence} between $U_i$ and $U_j$.
However, an alternative parametrization is available that accommodates {the independence} between the variables; see equation (\ref{eq:tri_density_c}) below.}
A very appealing aspect from both a tractability and computational viewpoint is the closed form of the density which does not include any integrals or infinite sums, unlike most existing distributions on the three-dimensional torus.

Note that condition~\eqref{conditions} can be re-expressed under simpler form as $ \rho_{ij}^2 \rho_{i k}^2 >  (|\rho_{ij}| + |\rho_{i k}|)$ under \eqref{identcond}.
The following result formally establishes the identifiability of the parameters in our model.
\begin{prop} \label{prop:identifiability}
The {TWCC} density (\ref{eq:tri_density}) has identifiable parameters.
\end{prop}

\begin{proof}
Let $\rhob = (\rho_{12},\rho_{13},\rho_{23})'$ and $\Tilde{\rhob} = (\tilde{\rho}_{12},\tilde{\rho}_{13},\tilde{\rho}_{23})'$ be two different points in the constrained parameter space of the proposed family (\ref{eq:tri_density}) with $\rho_{12}\rho_{13}\rho_{23}= \tilde{\rho}_{12}\tilde{\rho}_{13}\tilde{\rho}_{23} = \beta$.
Assume that $\rhob$ and $\tilde{\rhob}$ represent the same distribution, namely, $t(\ub; \rhob) = t(\ub;\tilde{\rhob})$ for any $\ub = (u_1,u_2,u_3)' \in [0,2\pi)^3$, where $t$ is the density (\ref{eq:tri_density}).
This implies that there exist real-valued constants $C$ and $D$ such that, for any $\ub \in [0,2\pi)^3$,
\begin{align*}
	\lefteqn{ \rho_{12} \cos (u_1 - u_2) + \rho_{13} \cos (u_1 - u_3) + \rho_{23} \cos (u_2 - u_3) } \hspace{1cm} \\
	& = D+C \left\{ \tilde{\rho}_{12} \cos (u_1 - u_2) + \tilde{\rho}_{13} \cos (u_1 - u_3) + \tilde{\rho}_{23} \cos (u_2 - u_3) \right\}.
\end{align*}
Choosing $(u_1,u_2,u_3)'$ equal to $(\pi/2,\pi/2,0)'$, $(\pi/2,0,\pi/2)'$, and $(0,\pi/2,\pi/2)'$, respectively, yields, $\rho_{ij}=D+C\tilde{\rho_{ij}}$ for each couple $(i,j)\in\{(1,2), (1,3), (2,3)\}$. Moreover, $(u_1,u_2,u_3)'=(0,0,0)'$ gives $\rho_{12}+\rho_{13}+\rho_{23}=D+C(\tilde{\rho_{12}}+\tilde{\rho_{13}}+\tilde{\rho_{23}})$. Summing the first three equalities and subtracting the last entails $0=2D$ and hence $D=0$, leading to
\begin{align*}
	\lefteqn{ \rho_{12} \cos (u_1 - u_2) + \rho_{13} \cos (u_1 - u_3) + \rho_{23} \cos (u_2 - u_3) } \hspace{1cm} \\
	& = C \left\{ \tilde{\rho}_{12} \cos (u_1 - u_2) + \tilde{\rho}_{13} \cos (u_1 - u_3) + \tilde{\rho}_{23} \cos (u_2 - u_3) \right\}.
\end{align*}
Then it follows that, for any $(u_1,u_2,u_3)$,
$$
(\rho_{12} - C \tilde{\rho}_{12}) \cos (u_1-u_2) + (\rho_{13} - C \tilde{\rho}_{13}) \cos (u_1-u_3) + (\rho_{23} - C \tilde{\rho}_{23}) \cos (u_2-u_3) = 0.
$$
Thus
$$
\rho_{12} = C \tilde{\rho}_{12}, \ \rho_{13} = C \tilde{\rho}_{13}, \ \rho_{23} = C \tilde{\rho}_{23}. \label{eq:rho}
$$
Using this equation and the assumption $\rho_{12} \rho_{13} \rho_{23} = \tilde{\rho}_{12} \tilde{\rho}_{13} \tilde{\rho}_{23} = \beta$, we have
$$
\beta = \rho_{12} \rho_{13} \rho_{23} = C^3 \tilde{\rho}_{12} \tilde{\rho}_{13} \tilde{\rho}_{23} = C^3 \beta. 
$$
Thus we have $C=1$.
This implies $\rhob= \tilde{\rhob}$, which is contradictory to the assumption $\rhob$ and $\tilde{\rhob}$ are two different points. 
\end{proof}


Next, we show that our model given by {TWCC} is indeed a copula for trivariate circular data by establishing in Theorem~\ref{thm:marginals} that its univariate marginals are circular uniform distributions.
Prior to this, we will show that the bivariate marginals of our new model are bivariate wrapped Cauchy-type {copulas} as they are of the form~\eqref{eq:wc_wj}.
\begin{thm} \label{thm:marginals}
Let a trivariate circular random vector $(U_1,U_2,U_3)'$ follow the {TWCC} with density  (\ref{eq:tri_density}).
Then the following hold for the marginal distributions of $(U_1,U_2,U_3)'$:
\begin{enumerate}
	\item[(i)] The marginal distribution of $(U_i, U_j)'$ is of the form~\eqref{eq:wc_wj} with density
\begin{equation}
t_2(u_i,u_j; \phi_{ij}) = \frac{1}{4\pi^2} \frac{ | 1-\phi_{ij}^2 | }{1+\phi_{ij}^2 - 2 \phi_{ij} \cos (u_i-u_j)}, \quad 0 \leq u_i,u_j < 2\pi, \label{eq:marginal_density}
\end{equation}
where
\begin{equation}\label{thephi}
{\phi_{ij} =  \frac{1}{ 2 \rho_{ij} } \left\{ \frac{\rho_{i k} \rho_{j k}}{\rho_{ij}} - \frac{\rho_{ij} \rho_{i k}}{\rho_{j k}} - \frac{\rho_{ij} \rho_{j k} }{ \rho_{i k} } - (2\pi)^3 c_2 \right\} ,}
\end{equation} and $c_2$ is as in (\ref{eq:c2}).

\item[(ii)] The marginal distribution of $U_i$ is the uniform distribution on the circle with density 
$$
t_1(u_i)=\frac{1}{2\pi}, \quad 0 \leq u_i < 2\pi.
$$
Therefore the distribution of $(U_1, U_2, U_3)'$ is a copula for trivariate circular data.
\end{enumerate} 
\end{thm}

\begin{proof}
Without loss of generality, we prove the case $(i,j)=(1,2)$.
The other cases can be shown in a similar manner.
	\begin{enumerate}
		\item[(i)] 
		It follows from the equation (\ref{eq:marginal_proof}) that the marginal density of $(U_1,U_2)'$ can be expressed as
		\begin{align*}
			 t_2(u_1, u_2; \phi_{12}) & = \int_0^{2\pi}  t(\ub; \rhob) d u_3 \nonumber \\
			& = \frac{2 \pi \cdot c_2}{|\rho_{12} \rho_{23} / \rho_{13} +  \rho_{12} \rho_{13} / \rho_{23} -\rho_{13} \rho_{23}/ \rho_{12} + 2 \rho_{12} \cos (u_1-u_2)|}.
		\end{align*}
		This marginal density can be rewritten as
  $$
   t_2(u_1, u_2; \phi_{12}) = \frac{1}{(2\pi)^2} \frac{|1-\phi_{12}^2|}{1+\phi_{12}^2 - 2 \phi_{12} \cos (u_1-u_2)}.
  $$
		The value of $\phi_{12}$ can be obtained as a solution to the equations $1+\phi_{12}^2=C \cdot ( \rho_{12} \rho_{23} / \rho_{13} + \rho_{12} \rho_{13} / \rho_{23}- \rho_{13} \rho_{23}/ \rho_{12}  )$ and $-2\phi_{12} = 2 C \rho_{12} $ for some $C \neq 0$.
  Specifically, since these equations imply 
  \begin{equation}
   - \frac{1+\phi_{12}^2}{2\phi_{12}} = \frac{ c_3 }{ 2\rho_{12}}, \label{eq:delta_12}
  \end{equation}
  where $c_3=\rho_{12} \rho_{23} / \rho_{13} + \rho_{12} \rho_{13} / \rho_{23}- \rho_{13} \rho_{23}/ \rho_{12}$, the solutions of this equation, say $\tilde{\phi}_{12}$, are given by
  \begin{align*}
  \tilde{\phi}_{12} & = \frac{-c_3 \pm \{ c_3^2 - 4 \rho_{12}^2 \}^{1/2} }{2 \rho_{12}} \\
  & = \frac{ \rho_{13} \rho_{23}/ \rho_{12} - \rho_{12} \rho_{23} / \rho_{13} - \rho_{12} \rho_{13} / \rho_{23} \pm (2\pi)^3 c_2 }{2 \rho_{12}}.
  \end{align*}
  Denote these solutions by $\tilde{\phi}_{12}^{+} = \{-c_3 + \{ c_3^2 - 4 \rho_{12}^2 \}^{1/2} \}/(2 \rho_{12}) $ and $\tilde{\phi}_{12}^{-} = \{-c_3 - \{ c_3^2 - 4 \rho_{12}^2 \}^{1/2} \}/(2 \rho_{12})$.
  Then by definition of $\phi_{12}$ it is straightforward to see $\tilde{\phi}_{12}^{-} = \phi_{12}$,  $\tilde{\phi}_{12}^+ \tilde{\phi}_{12}^-=1$ and $\tilde{\phi}_{12}^{+}=1/\phi_{12}$.
  Note that $ t_2(u_1, u_2; \phi_{12})$ in equation (\ref{eq:marginal_density}) with the parameter $\phi_{12} (\neq 0)$ satisfies
  \begin{equation}
  t_2(u_1,u_2 ; \phi_{12} ) = t_2(u_1,u_2 ; 1/\phi_{12} ), \quad  0 \leq u_1,u_2 < 2\pi.
  \label{eq:cauchy_equiv}
  \end{equation}
  This expression implies that the two solutions of (\ref{eq:delta_12}), i.e., $\tilde{\phi}_{12}^{+}$ and $\tilde{\phi}_{12}^{-}$, correspond to the parameters of the same distribution.
  Using the parameters derived from the solution $\tilde{\phi}_{12}^-$, we have
  $$
 t_2(u_1, u_2; \phi_{12}) \propto \left\{ 1+\phi_{12}^2 - 2 \phi_{12} \cos (u_1-u_2) \right\}^{-1}.
  $$
  Since the functional form of this marginal density is essentially the same as that of (\ref{eq:wc_wj}), it follows that the marginal density is given by (\ref{eq:marginal_density}) with $(i,j)=(1,2)$.
		
		\item[(ii)] It immediately follows from equation (\ref{eq:uniform}) that the marginal distribution of $U_1$ is the uniform distribution on the circle.
	\end{enumerate}
\end{proof}

{Note that the $\phi_{ij}$ are  invariant under $\rho_{ij}$ replaced by $c\rho_{ij}$ for some constant $c\in\R^+$, implying that they do not depend on any identifiability condition on $\rho_{ij}$.} We draw the reader's attention to the fact that the constant $c_1$ has to be of the form (\ref{eq:c1}), which guarantees that the bivariate marginal distribution belongs to the bivariate wrapped Cauchy-type family. See  equality (\ref{eq:marginal_proof}).

\section{Properties of the TWC copula}\label{sec:Props}

We will investigate distinct properties of our new copula distribution {TWCC($\boldsymbol{\rho}$)}. In order to do so, it is often convenient to express its density using complex variables.
Let  $(U_1,U_2,U_3)'$ have the {TWCC} density (\ref{eq:tri_density}), and assume that $(Z_1,Z_2,Z_3)' = (e^{{\rm i} U_1}, e^{{\rm i} U_2}, e^{{\rm i} U_3})'$.
Then some simple algebra yields that the {TWCC} density  can be expressed as
\begin{equation}
tc(\zb; \phib) = \frac{1}{(2\pi)^3} \frac{ \{ \phi_1^4 + \phi_2^4 + \phi_3^4 - 2 \phi_1^2 \phi_2^2 - 2 \phi_1^2 \phi_3^2 - 2\phi_2^2 \phi_3^2 \}^{1/2} }{|\phi_1 z_1 + \phi_2 z_2 + \phi_3 z_3|^2}, \label{eq:tri_density_c}
\end{equation}
where $\zb = (z_1,z_2,z_3)' \in \Omega = \{ z \in \mathbb{C} \, ; \, |z|=1 \}$ is the unit circle in the complex plane and $\phib = (\phi_1, \phi_2, \phi_3)'$ and $\phi_i$ is defined as in (\ref{eqrelat}).
We call this the complex {TWCC}. 
Note that, in terms of the original parameters, we have $\rho_{12}=\phi_1\phi_2$, $\rho_{13}=\phi_1\phi_3$ and $\rho_{23}=\phi_2\phi_3$.
The inequality \eqref{conditions} on the parameters in Theorem \ref{thm:density} then simplifies to $|\phi_i| > |\phi_j| + |\phi_{k}|$ for $(i,j,k)$ a certain permutation of $(1,2,3)$.
In fact, this is equivalent to the condition that the denominator of (\ref{eq:tri_density_c}) satisfies $\phi_1 z_1 + \phi_2 z_2 + \phi_3 z_3 \neq 0$ for all $(z_1,z_2,z_3)'$, see Lemma~\ref{lem_App} in \ref{sec:technical lemmas} for a statement and proof.
Equivalently, the condition on the parameter $|\rho_{j k}| < |\rho_{ij} \rho_{i k}| / ( |\rho_{ij}| + |\rho_{i k}|)$ for some $(i,j,k)$ is necessary to guarantee the boundedness of the {TWCC($\boldsymbol{\rho}$)} for all $(u_1,u_2,u_3)'$. 
The identifiability constraint~\eqref{identcond} turns into $(\phi_1\phi_2\phi_3)^2=1.$
Note that the parameter space of the complex TWCC must satisfy $|\phi_i| > |\phi_j| + |\phi_{k}|$.

If the parameter space is extended to include $(\phi_i \phi_j)^2 = 1$ and $\phi_k = 0$, the complex TWCC includes the case where $Z_k$ is independent of both $Z_i$ and $Z_j$.
The additional assumptions $\phi_i = 1 + \varepsilon$ and $\phi_j = -(1 + \varepsilon)^{-1}$ for small $\varepsilon > 0$ imply strong positive dependence between $Z_i$ and $Z_j$.
As $\varepsilon \rightarrow \infty$, the complex TWCC converges to the circular equivalent of the comonotonic copula of $(Z_i, Z_j)$, that is, $P(Z_i = Z_j) = 1$.

The parameter space can also be extended to include $\phi_i = 1$ and $\phi_j = \phi_k = 0$, which corresponds to the independence copula.

In the next sections, we will study and discuss conditional distributions, random variate generation and modality; {in Sections \ref{Sup_trigmom}--\ref{sec:correlation_coef} of the Supplementary Material we further investigate trigonometric moments, correlation coefficients, the shapes of our distribution, and limiting cases.}

\subsection{Conditional distributions and regression}

In this subsection we consider the conditional distributions of the  model {TWCC($\boldsymbol{\rho}$)}. As we will show, just like the {bivariate and univariate} marginal distributions, all conditional distributions belong to well-known families, namely to wrapped Cauchy distributions on the circle and to the Kato--Pewsey distribution on the torus. This   property  also  highlights  the versatility of the {TWCC}. 

\begin{thm} \label{thm:conditionals}
	Let $(U_1,U_2,U_3)'$ be a trivariate random vector having the distribution {TWCC($\boldsymbol{\rho}$)}.
	Then the conditional distributions of $(U_1,U_2,U_3)'$ are given below.
	\begin{enumerate}
    		\item[(i)] The conditional distribution of $(U_i,U_j)'$ given $U_{k}=u_{k}$ is a reparametrized version of the distribution of \cite{KP15} with density
		{\begin{align}
		\lefteqn{ t_{2|1}(u_i,u_j|u_{k}; \rhob) } \hspace{0.7cm} \nonumber \\
		 = \ &  2 \pi c_2 \, \Bigl[ c_1 + 2 \left\{ \rho_{ij} \cos (u_i - u_j) + \rho_{i k} \cos (u_i - u_{k}) + \rho_{j k} \cos (u_j - u_{k}) \right\} \Bigr]^{-1} \label{eq:conditional_1st} \\
		 	= \ & 2\pi c_2 \, \Bigl[ c_1 + 2 \{  \rho_{i k} \cos (u_i - u_{k}) + \rho_{j k} \cos (u_j - u_{k}) \nonumber \\
		 & + \rho_{ij} \cos (u_i - u_{k}) \cos (u_j - u_{k}) + \rho_{ij} \sin (u_i - u_{k}) \sin (u_j - u_{k}) \} \Bigr]^{-1}, \label{eq:conditional_2nd}\\
		 & \hspace{9cm} 0 \leq u_i,u_j < 2\pi . \nonumber
		\end{align}}
        
	\item[(ii)] The conditional distribution of $U_i$ given $U_j=u_j$ is the wrapped Cauchy distribution with density
	\begin{equation}
		{ t_{1|1}(u_i | u_j; \phi_{ij}) = \frac{1}{2\pi} \frac{ | 1-\phi_{ij}^2 | }{1+\phi_{ij}^2 - 2 \phi_{ij} \cos (u_i- u_j ) }, \quad 0 \leq u_i < 2\pi, }\label{eq:conditional_density1}
	\end{equation}
	where $\phi_{ij}$ is as in  \eqref{thephi}.

    	\item[(iii)] 
	The conditional distribution of $U_i$ given $(U_j,U_{k})' =(u_j,u_{k})'$ is the wrapped Cauchy distribution with density
	\begin{equation}
	t_{1|2}(u_i | u_j,u_{k}; \eta_{i|jk} , \delta_{i|jk}) = \frac{1}{2 \pi} \frac{ | 1-\delta_{i|j k}^2 | }{1+\delta_{i|j k}^2 - 2 \delta_{i|j k} \cos (u_i- \eta_{i|j k})}, \label{eq:conditional_density2}
	\end{equation}
    where $0 \leq u_i < 2\pi,$ and, for $\phi_{i|j k} = - \rho_{j k} (\rho_{ik}^{-1} e^{{\rm i} u_j} + \rho_{i j}^{-1} e^{{\rm i} u_{k}}) $, 
    \begin{align}
        & \eta_{i|j k} =  \arg (\phi_{i|j k}) \quad \text{and} \quad \label{eq: params_conditional} \delta_{i|j k} =|\phi_{i | j k}|.
    \end{align}
	
	\end{enumerate}
\end{thm}


\begin{proof}
	Without loss of generality, we consider the case $(i,j,k)=(1,2,3)$.
	\begin{enumerate}
		\item[(i)] It is straightforward to derive the first expression of the conditional density (\ref{eq:conditional_1st}) from the equation $t_{2|1}(u_1,u_2 | u_3; \rhob) = t(\ub; \rhob) / t_1(u_3)$, where $t(\ub; \rhob)$ is the trivariate density (\ref{eq:tri_density}) and $t_1(u_3)$ is the density of $U_3$, namely, the circular uniform density (see Theorem \ref{thm:marginals}(ii)).
		It follows from equation (2) of \cite{KP15} that the second expression of the conditional density (\ref{eq:conditional_2nd}) has the same functional form apart from parametrization.
		\item[(ii)] Theorem \ref{thm:marginals} implies that the density of $(U_1,U_2)'$ is given by (\ref{eq:marginal_density}) and the density of $U_2$ is the circular uniform density.
		Then it follows from the expression $t_{1|1}(u_1|u_2;\phi_{12})=t_2(u_1,u_2; \phi_{12})/t_1(u_2)$ that the conditional density of $U_1$ given $U_2=u_2$ is the wrapped Cauchy density (\ref{eq:conditional_density1}).  
		\item[(iii)] Using the complex expression of the density (\ref{eq:tri_density_c}), the conditional density of $Z_1$ given $(Z_2,Z_3)' = (z_2,z_3)'$ is of the form
			$$
			tc_{1|2}(z_1|z_2,z_3) \propto \left| z_1 + \frac{\phi_2 z_2 + \phi_3 z_3}{\phi_1} \right|^{-2}, \quad z_1 \in \Omega.
			$$
			Note that the density of the wrapped Cauchy distribution can be expressed as
			$$
			f(z) = \frac{1}{2\pi} \frac{|1-\delta^2|}{|z-\delta e^{{\rm i} \eta}|^2}, \quad z \in \Omega,
			$$
			where $ \eta \in [0, 2\pi) $ is the location parameter and $\delta \geq 0$ is the concentration parameter (see \cite{MCC96}).
			It follows that the conditional of $Z_1$ given $(Z_2,Z_3)' = (z_2,z_3)'$ is the wrapped Cauchy distribution with the location parameter $\arg(\phi_{1|23})$ and concentration parameter $|\phi_{1|23}|$, where $\phi_{1|23}=- \phi_1^{-1} ( \phi_2 z_2 + \phi_3 z_3) = - \rho_{2 3} (\rho_{13}^{-1} e^{{\rm i} u_2} + \rho_{12}^{-1} e^{{\rm i} u_3})$. 
	\end{enumerate}
\end{proof}


As this theorem shows, the univariate conditionals in Theorem \ref{thm:conditionals}(ii) and (iii) have the wrapped Cauchy distributions.
Note that the univariate conditional given in Theorem~\ref{thm:conditionals}(ii) does not follow the wrapped Cauchy in general if $c_1$ is not defined as in (\ref{eq:c1}).
The bivariate conditional in Theorem \ref{thm:conditionals}(i) has various tractable properties as discussed in \cite{KP15}.

The well-known form of the conditional distributions paves the way for regression purposes with one or two angular dependent variables and/or one or two angular regressors. Indeed, if we wish to predict one angular component based on two angular components, then from {Theorem} \ref{thm:conditionals}(iii) we find that the mean direction and circular variance of $U_i$ given $(U_j,U_k)'$ are simply $\eta_{i|j k}$ and $1-\delta_{i|j k}$, respectively. Similarly, if we wish to predict two angular components based on a third angular component, then from {Theorem} \ref{thm:conditionals}(i) we see that the toroidal mean and variance of $(U_i,U_j)'$ given $U_k$ can be calculated in the same way as in Section 2.5 of \cite{KP15}. Note that circular-circular regression {of, for example, $U_i$ given $U_j$} can also be obtained in a straightforward way from {Theorem} \ref{thm:conditionals}(ii). These properties  are not explored here  but we emphasize that they can be {used} in practice easily.


\subsection{Random variate generation} \label{sec:random_variate}

The fact that all the marginal and conditional distributions belong to existing tractable families lays the foundations for random variate generation. Indeed, random variates from the proposed trivariate model {TWCC($\boldsymbol{\rho}$)} can be efficiently generated from uniform random variates on $(0,1)$.

\begin{thm} \label{thm:random}
	The following algorithm generates random variates from the distribution {TWCC($\boldsymbol{\rho}$)} without rejection.
	\begin{enumerate}[Step 1.] 
		\item[Step 1.] Generate uniform $(0,1)$ random variates $\omega_1,\omega_2$ and $\omega_3$.
		\item[Step 2.] Compute
		\begin{align*}
		u_1 & =2\pi \omega_1, \\
        u_2 & = u_1 + \arg (\phi_{12}) + 2 \arctan \left[ \left(\frac{1-|\phi_{12}|}{1+|\phi_{12}|}\right) \tan \left\{ \pi (\omega_2-0.5) \right\} \right], \\
		u_3 & = \eta_{3|12} +  2 \arctan \left[ \left(\frac{1-\delta_{3|12}}{1+\delta_{3|12}}\right) \tan \left\{ \pi (\omega_3-0.5) \right\} \right],
		\end{align*}
	where $\phi_{12}$ {is as in \eqref{thephi}} and $(\eta_{3|12}, \delta_{3|12})$ {are as in \eqref{eq: params_conditional}}.
		\item[Step 3.] Output $(u_1,u_2,u_3)'$ as the random variate from the distribution {TWCC($\boldsymbol{\rho}$)}.
	\end{enumerate}
\end{thm}

\begin{proof}
	The trivariate density (\ref{eq:tri_density}) can be decomposed as
	$$
	t(\ub; \rhob) = t_{1|2}(u_3|u_1,u_2; \eta_{3|12} , \delta_{3|12}) t_{1|1}(u_2|u_1; \phi_{12}) t_1(u_1).
	$$
	Theorems \ref{thm:marginals} and \ref{thm:conditionals} imply that $t_{1|2}(u_3|u_1,u_2; \eta_{3|12}, \delta_{3|12})$ is the wrapped Cauchy density (\ref{eq:conditional_density2}), $t_{1|1}(u_2|u_1; \phi_{12})$ is also the wrapped Cauchy density (\ref{eq:conditional_density1}), and $t_1(u_1)$ is the circular uniform density.
	This expression implies that the random variate generation from the proposed trivariate distribution (\ref{eq:tri_density}) is equivalent to that from the circular uniform and wrapped Cauchy distributions.
	
	It is straightforward to see that $u_1$ computed in Step 2 is a random variate from the circular uniform distribution.
	In order to generate random variates $u_2$ and $u_3$ from the conditional wrapped Cauchy distributions, we apply the following result:
	if a random variable $U$ follows the circular uniform distribution on $(-\pi,\pi)$, then the random variable defined by
	$$
	\Theta =  \eta + 2 \arctan \left\{ \left(\frac{1-\delta}{1+\delta}\right) \tan \left( \frac{U}{2} \right) \right\}
	$$ 
	has the wrapped Cauchy distribution with location parameter $\eta \in [0,2\pi)$ and concentration parameter $\delta \in \mathbb{R}^+ \setminus \{  1\}$ with density
	$$
	f(\theta) = \frac{1}{2\pi} \frac{|1-\delta^2|}{1+\delta^2-2 \delta \cos (\theta - \eta)}, \quad 0 \leq \theta < 2\pi.
	$$
	See \cite{MCC96} and \cite{KP15} for details.
	Using this result, it is straightforward to see that $u_2$ and $u_3$ computed in Step 2 are variates from the conditional wrapped Cauchy distributions with location parameters $u_1+\arg (\phi_{12} )$ and $\eta_{3|12}$ and concentration parameters $|\phi_{12}|$ and $\delta_{3|12}$, respectively.
\end{proof}

{We note} the simplicity and efficiency of the algorithm in which a variate from the proposed distribution can be generated through a transformation of three uniform random variates without rejection.

{


	
}

\subsection{Modality}

In this section we investigate the modes of {TWCC($\boldsymbol{\rho}$)}.

\begin{thm} \label{thm:modes}
For $\rho_{ij}>0$ and $\rho_{i k},\rho_{j k}<0$, the modes of the density {TWCC($\boldsymbol{\rho}$)} are given by
	{\begin{equation}
	\begin{array}{cll}
		\mbox{(i)} & u_i=u_j=u_{k} &  \mbox{if \ } |\rho_{ij}| < |\rho_{i k} \rho_{j k}|/(|\rho_{i k}| + |\rho_{j k}|), \\
		\mbox{(ii)} & u_i=u_j+\pi = u_{k}+\pi \quad &  \mbox{if \ } |\rho_{i k}| < |\rho_{i j} \rho_{j k}|/(|\rho_{i j}| + |\rho_{j k}|),		
	\end{array} \label{eq:modes}
\end{equation} }
and the antimodes of the density {TWCC($\boldsymbol{\rho}$)} are given by
\begin{equation}
u_i=u_j=u_{k}+\pi. \label{eq:antimodes}
\end{equation}
If $\rho_{ij},\rho_{i k}, \rho_{j k} > 0$, then $u_{k}$ in the modes (\ref{eq:modes}) and antimodes (\ref{eq:antimodes}) is replaced by $u_{k} + \pi$.
\end{thm}

\begin{proof}
	First we consider the case $\rho_{ij},\rho_{ik},\rho_{j k}>0$.
	Without loss of generality, assume $(i,j,k)=(1,2,3)$.
	It is clear that the antimodes of the density (\ref{eq:tri_density}) are $u_1=u_2=u_3$ because $\cos (u_i-u_j)$ $(1 \leq i<j \leq 3)$ is maximized at $u_i=u_j$.
	In order to derive the modes of the density (\ref{eq:tri_density}), we first note that the density (\ref{eq:tri_density}) can be expressed as
    \begin{align*}
	t(\ub; \rhob) \propto & \Bigl[ c_1 + 2 \rho_{12} \cos (u_1 - u_2)  \\
	& + 2 \{ \rho_{13}^2 + \rho_{23}^2 + 2 \rho_{13} \rho_{23} \cos (u_1 - u_2) \}^{1/2} \cos (u_3 - \tilde{\eta}_{3|12}) \Bigr]^{-1},
	\end{align*}
	where $\tilde{\eta}_{3|12} =  \arg \{ \rho_{13} \cos u_1 + \rho_{23} \cos u_2 + {\rm i} (\rho_{13} \sin u_1 + \rho_{23} \sin u_2) \}$.
	It immediately follows from this expression that the density (\ref{eq:tri_density}) is maximized at $u_3=\tilde{\eta}_{3|12}+\pi$. 
	Then the maximization of the density (\ref{eq:tri_density}) reduces to the minimization of its functional part
	$$
	B(x) = \rho_{12} x -  \{ \rho_{13}^2 + \rho_{23}^2 + 2 \rho_{13} \rho_{23} x \}^{1/2},
	$$
	where $x = \cos (u_1-u_2) \in [-1,1]$.
	The first derivative of this function is
	\begin{equation}\label{deriv}
	\frac{d}{dx} B(x) = \rho_{12} - \rho_{13} \rho_{23} (\rho_{13}^2+\rho_{23}^2+2 \rho_{13} \rho_{23} x)^{-1/2}.
	\end{equation}
	Now it is straightforward to see that this derivative can be upper bounded by $\rho_{12}-\rho_{13}\rho_{23}/(\rho_{13}+\rho_{23})$ (by replacing $x$ with 1 in~\eqref{deriv}), and that the derivative is thus negative if $\rho_{12}<\rho_{13}\rho_{23}/(\rho_{13}+\rho_{23})$, leading to a minimization of $B(x)$
 at $x=1$ if $ \rho_{12} < \rho_{13} \rho_{23} / ( \rho_{13} + \rho_{23} )$. By using similar arguments, one can show that $B(x)$ is minimized at $x=-1$ if $ \rho_{12} > \rho_{13} \rho_{23}  / |\rho_{13} - \rho_{23}| $.
	It then follows that the modes of the density (\ref{eq:tri_density}) are given at $u_1=u_2$ for $ \rho_{12} < \rho_{13} \rho_{23}  / (\rho_{13} + \rho_{23} )$ and at $u_1=u_2+\pi$ for $ \rho_{12} > \rho_{13} \rho_{23}  / |\rho_{13} - \rho_{23}|$.
	Finally, $u_3=\tilde{\eta}_{3|12}+\pi$ implies that $u_3= u_1+\pi $ if  $ \rho_{12} < \rho_{13} \rho_{23}  / (\rho_{13} + \rho_{23} ) $.
	If $ \rho_{12} > \rho_{13} \rho_{23}  / |\rho_{13} - \rho_{23}|$, then we have $u_3=u_1$ for $\rho_{13}-\rho_{23}<0$.
	Noticing that $ \rho_{12} > \rho_{13} \rho_{23}  / |\rho_{13} - \rho_{23}|$ and $\rho_{13} - \rho_{23} <0$ imply $\rho_{13}<\rho_{12} \rho_{23}/(\rho_{12} + \rho_{23})$, we obtain the modes of the density (\ref{eq:tri_density}) for the case $(i,j,k)=(1,2,3).$
	If $ \rho_{12} > \rho_{13} \rho_{23}  / |\rho_{13} - \rho_{23}|$ and $\rho_{13}-\rho_{23} > 0$, the modes of the density (\ref{eq:tri_density}) can be obtained by setting $(i,j,k)=(2,1,3)$ rather than $(i,j,k)=(1,2,3)$ due to the symmetry of the density (\ref{eq:tri_density}) with respect to the permutation of $(\rho_{ik},\rho_{j k})$.
	
	The modes and antimodes of the density (\ref{eq:tri_density}) for the case {$\rho_{ij}>0$ and $\rho_{ik},\rho_{j k}<0$} can be obtained in the same manner by noticing the straightforward equality $ t(u_i,u_j,u_k; \rho_{ij}, \rho_{ik},\rho_{jk}) 
		 = t(u_i, u_j, u_k+\pi; \rho_{ij},-\rho_{ik},-\rho_{jk})$.
    \end{proof}

Since the conditions in \eqref{eq:modes} {are the same as \eqref{conditions} in Theorem~\ref{thm:density}}, one of them is always satisfied, which makes the result very strong as it shows that our trivariate copula has modes lying on a single line. Moreover, the modes of {TWCC($\boldsymbol{\rho}$)} have the most natural form in the case (i) of equation (\ref{eq:modes}), and very simple forms in the other cases. This  property simplifies working with mixtures.

\section{Parameter Estimation} \label{sec: parameter estimation}

In this section we consider  maximum likelihood estimation, and {method of moment estimation is discussed in Section \ref{sec:mom} of the Supplementary Material.}
Throughout this section, let $\{ (u_{1m},u_{2m},u_{3m})' \}_{m=1}^n $ be a realization of a random sample $\{ (U_{1m},U_{2m},U_{3m})' \}_{m=1}^n $ from the distribution {TWCC($\rhob$)}.

For the original model (\ref{eq:tri_density}), the likelihood function for $\{ \ub_m \}_{m=1}^n = \{ (u_{1m},u_{2m},u_{3m})' \}_{m=1}^n$ is given by
\begin{align} \label{eq:likelihood}
\log L (\rhob) = \ \log \prod_{m=1}^n t(\ub_m)
= \ n \log c_2 - \sum_{m=1}^n \log F_m,
\end{align}
where $F_m = c_1 + 2 \{ \rho_{12} \cos (u_{1m} - u_{2m}) + \rho_{13} \cos (u_{1m} - u_{3m}) + \rho_{23} \cos (u_{2m} - u_{3m}) \}$. Its score function is
\begin{align*}
& \frac{\partial}{\partial \rho_{ij}} \log L (\rhob)
	 = n \frac{ \frac{\partial c_2}{\partial \rho_{ij}}}{c_2} - \sum_{m=1}^n \frac{ \frac{\partial c_1}{\partial \rho_{ij}} + 2 \cos (u_{im} - u_{jm})}{F_m}
\end{align*}
where
\begin{align*}
    & \frac{\partial c_1}{\partial \rho_{ij}}
    = \frac{\rho_{ik}}{\rho_{jk}} + \frac{\rho_{jk}}{\rho_{ik}} - 
    \frac{\rho_{jk}\rho_{ik}}{\rho_{ij}^2}, \\
    & \frac{\partial c_2}{\partial \rho_{ij}}
    = \frac{1}{(2\pi)^3} \left\{ \rho_{ij} \left( \left( \frac{\rho_{ik}}{\rho_{jk}} \right)^2 + \left( \frac{\rho_{jk}}{\rho_{ik}} \right)^2 \right) -  \frac{(\rho_{jk}\rho_{ik})^2}{\rho_{ij}^3} - 2 \rho_{ij} \right\} \left((2\pi)^3 c_2\right)^{-1}.
\end{align*}
{The maximum likelihood estimates $\hat{\rho}_{12}, \hat{\rho}_{13}, \hat{\rho}_{23}$ are obtained by equating the score function to zero and numerically solving the system of three equations, see below. 
For the reader's convenience, we provide {in Section~\ref{sec:fisher_information} in the Supplementary Material} the associated expected Fisher information matrix.}

Although the proposed density (\ref{eq:tri_density}) has a simple and closed form, it is essential to keep  the parameter constraints \eqref{conditions} and \eqref{identcond} in mind. The simplified constraint
$ |\rho_{ij}| + |\rho_{ik}| <  |\rho_{ij}|^2 |\rho_{ik}|^2$ can further be turned into
$$
| \rho_{ij} | > \frac{1 + \{ 1 + 4|\rho_{ik}|^3 \}^{1/2} }{2 |\rho_{ik}|^2}
$$
by solving an inequality of the second order.
Using the expression 
$$
\zeta_{ij} = \frac{1 + \{ 1 + 4|\rho_{ik}|^3 \}^{1/2} }{2 |\rho_{ik}|^2} \frac{1}{\rho_{ij}}, \quad \zeta_{ij} \in \left(-1,0\right)\cup\left(0,1\right),
$$
it is straightforward to see that the parameters of the  model (\ref{eq:tri_density}) can be expressed in terms of $\rho_{ik}$ and $\zeta_{ij}$ alone (remember that this holds for one choice of $i, j, k$).
The parameter space of this reparametrized model under the constraint 
 $|\rho_{j k}| < |\rho_{ij} \rho_{i k}| / ( |\rho_{ij}| + |\rho_{i k}|)$ and $\rho_{12} \rho_{13} \rho_{23} = 1$ is thus
$$
\tilde{\Omega} = \left\{ (\zeta_{ij}, \rho_{ik}) \, ; \, \zeta_{ij} \in \left(-1,0\right) \cup \left(0,1\right), \rho_{ik} \in \mathbb{R} \setminus \{0\} \right\}.
$$
For notational convenience, write $\log L (\rho_{12}, \rho_{13}, \rho_{23}) = \log L (\zeta_{ij}, \rho_{ik})$ if the log-likelihood function (\ref{eq:likelihood}) is represented in terms of $(\zeta_{ij}, \rho_{ik})$.
Then the maximum likelihood estimation for the proposed model can be carried out as follows:
\begin{enumerate}
\item[Step 1.] For $(j,k)$ successively being equal to $(1,2),(2,3),(3,1)$, obtain the following estimates
\begin{align*}
(\tilde{\zeta}_{ij}, \tilde{\rho}_{ik}) & = \argmax_{(\zeta_{ij}, \rho_{ik}) \in \tilde{\Omega}}  \log L(\zeta_{ij}, \rho_{ik}).
\end{align*}
\item[Step 2.] Among the three obtained maximized quantities, calculate
$$
(\hat{\zeta}_{i^*j^*} , \hat{\rho}_{i^*k^*}) =  \argmax_{(\tilde{\zeta}_{ij}, \tilde{\rho}_{ik})} \, \log L(\tilde{\zeta}_{ij}, \tilde{\rho}_{ik}).
$$
\item[Step 3.] Record the maximum likelihood estimate $(\hat{\rho}_{12},\hat{\rho}_{13},\hat{\rho}_{23})'$ as
$$
\hat{\rho}_{i^*k^*} = \hat{\rho}_{i^*k^*} , \quad 
\hat{\rho}_{i^*j^*} = \frac{1 + \{ 1 + 4|\hat{\rho}_{i^*k^*}|^3 \}^{1/2} }{2 |\hat{\rho}_{i^*k^*}|^2} \frac{1}{\hat{\zeta}_{i^*j^*}}, 
\quad \hat{\rho}_{j^*k^*} =\frac{1}{ \hat{\rho}_{i^*j^*} \hat{\rho}_{i^*k^*}},
$$
where $i^* \neq j^* \neq k^*$.

\end{enumerate}

The algorithm is repeated with different initial values, to make sure that the global maximum is achieved. 
The initial values for the parameters $\zeta_{ij}$ and $\rho_{ik}$ are uniformly chosen from the intervals they are allowed to take values in {(where of course the infinite intervals for $\rho_{ik}$ are limited to a large maximal value)}. 
The function for calculating the ML estimates was written in the programming language \textsf{R}, using the optimizer solnp from the library Rsolnp \cite{Rsolnp} and is available in the GitHub repository
\url{https://github.com/Sophia-Loizidou/Trivariate-wrapped-Cauchy-copula}.
{The consistency of this approach is shown and its finite-sample performance is investigated by means of Monte Carlo simulations, which we provide in {Section~\ref{sec:simulations} in the Supplementary Material.}}

{

{ 
\section{Real-data analysis} \label{sec:real_data}


Predicting the 3D structure of proteins is a cutting-edge area of research in bioinformatics and computational biology. \cite{Mardia13} has provided a recent overview on statistical approaches to three major areas  in protein structural bioinformatics and listed the main challenges lying ahead. An essential challenge is the ability to jointly model the conformational angles $\phi,\psi$ (dihedral angles) and $\omega$ (torsion angle of the side chain). We will use our new distribution to precisely model those angles.

For the present data analysis, we consider position 55 at 2000 randomly selected times in the molecular dynamic trajectory of the SARS-CoV-2 spike domain from \cite{genna_sars-cov-2_2020}.
The position occurs in $\alpha$-helix throughout the trajectory. DPPS \cite{kabsch_dictionary_1983} is used to compute the secondary structure and \cite{cock_biopython_2009} to verify the chains. The parameters of our model are estimated using MLE as explained in Section~\ref{sec: parameter estimation}.

Figure \ref{fig:protein_contour_plots} shows the plots of the data of two of the angles given the third one. 
The values of the fixed angles in radians ($\phi = 1.93$, $\psi = 2.8$ and $\omega = 0$) are chosen to be the mode of the data, and only points that are within 0.1 radians of the selected value of the fixed angle are plotted.
In order to make the plots clearer, the range of values on each axis is not between $0$ and $2\pi$, like the traditional Ramachandran plot, but it is chosen such that both the contour plots and the points are visible (in summary, we have re-scaled the plots). 
In order to fit the data with our model, we shift them to have mean 0. The estimates of the copula parameters for this example are $\hat{\rho}_{12} = 0.611, \hat{\rho}_{13} = -1.31, \hat{\rho}_{23} = -1.25 $.
The contour plots in Figure \ref{fig:protein_contour_plots} correspond to our copula density \eqref{eq:tri_density}.
The observed data points are plotted on top of the contours. 
This gives a visual indication of how good our estimated model fits this protein dataset. We attract the reader's attention to the fact that our modality result from Theorem~\ref{thm:modes} implies that, conditionally on one fixed component, the conditional distribution of the other two is necessarily unimodal, which can nicely be seen from Figure~\ref{fig:protein_contour_plots}. 


\begin{figure}
\begin{subfigure}{.3\textwidth}
  \centering
  \includegraphics[width = 4.6cm, height = 3.5cm]{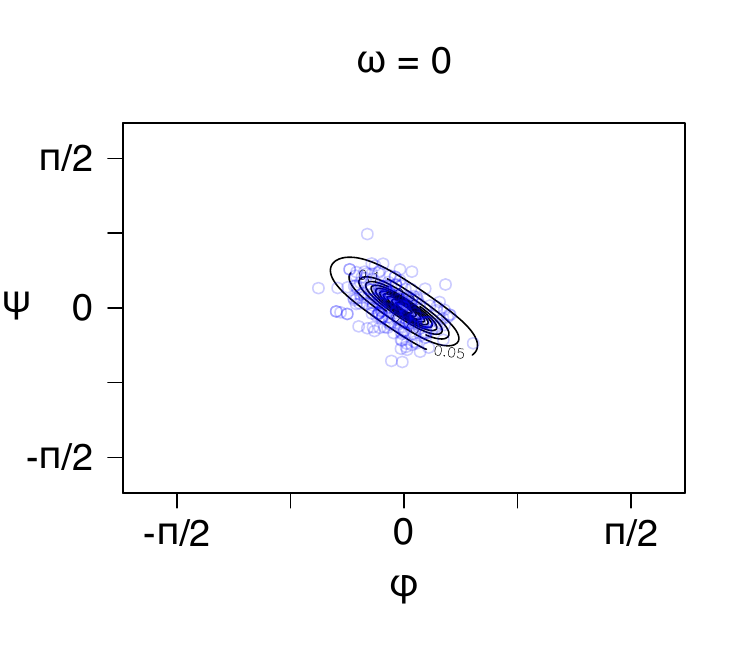}
  \caption{\hspace*{-0.7cm}}
  \label{fig:sfig1}
\end{subfigure}
\begin{subfigure}{.3\textwidth}
  \centering
  \includegraphics[width = 4.6cm, height = 3.5cm]{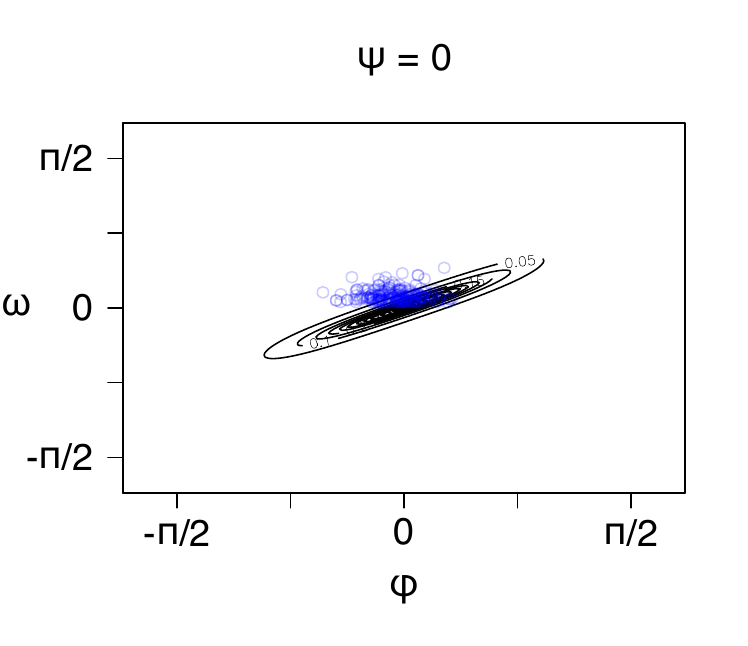}
  \caption{\hspace*{-0.7cm}}
  \label{fig:sfig2}
\end{subfigure}
\begin{subfigure}{.3\textwidth}
  \centering
  \includegraphics[width = 4.6cm, height = 3.5cm]{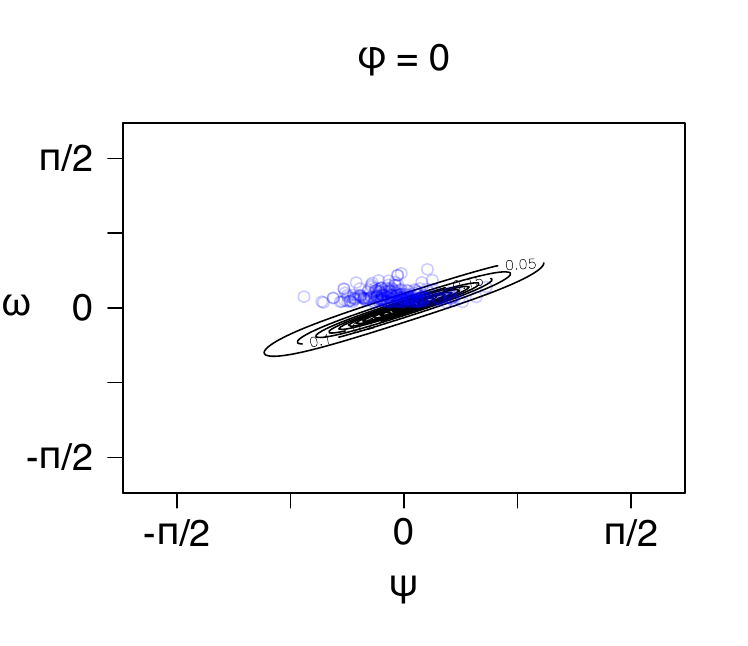}
  \caption{\hspace*{-0.7cm}}
  \label{fig:sfig3}
\end{subfigure}
\caption{Contour plots of the bivariate conditional density of the distribution \eqref{eq:tri_density}, with parameter values estimated by maximum likelihood. 
The  parameter estimates are $\hat{\rho}_{12} = 0.611, \hat{\rho}_{13} = -1.31, \hat{\rho}_{23} = -1.25 $. The values of the fixed angles (shown in each title) are chosen to be the mode of the data, and only points that are within 0.1 radians of the chosen value of the fixed angle are plotted. For illustration purposes, the three plots are portrayed on different scales.} \label{fig:protein_contour_plots} 
\label{fig:fig}
\end{figure}

\section{Discussion} \label{sec:conclusion}
In this paper we have proposed a new distribution on the three-dimensional torus, the trivariate wrapped Cauchy copula.
The marginal and conditional distributions of the copula are known distributions and random variate generation is simple and efficient through a transformation of uniform random variates without rejection. We have explicit expressions for the locations of the modes of our distribution.
Parameter estimation can be performed via maximum likelihood estimation. A natural question is how to extend our construction to any dimension $d$. We investigate such a proposal in the Supplementary Material. In future work, we intend to explore choices of marginal distributions which will extend the domain of potential applicability of our new copula. 



\textbf{Code availability:} The relevant code of the paper is available in \url{https://doi.org/10.5281/zenodo.15675162} as well as in \url{https://github.com/Sophia-Loizidou/Trivariate-wrapped-Cauchy-copula}.

\textbf{Acknowledgements:} The authors would like to thank Thomas Hamelryck and Ola R\o nning for the protein data. Kanti Mardia would like to thank the Leverhulme Trust for the Emeritus Fellowship.

\textbf{Funding:} Shogo Kato was supported by JSPS KAKENHI Grant Number 20K03759 and 25K15037.
Sophia Loizidou was supported by the grant PRIDE/21/16747448/MATHCODA from the Luxembourg National Research Fund.


\bibliographystyle{abbrv}
\bibliography{references_multcop}  






\clearpage
\begingroup
\renewcommand{\thepage}{S\arabic{page}} 
\setcounter{page}{1}                    
\renewcommand{\thesection}{S\arabic{section}} 
\setcounter{section}{0}
\renewcommand{\thefigure}{S\arabic{figure}}
\renewcommand{\thetable}{S\arabic{table}}
\renewcommand{\theequation}{S\arabic{equation}}
\renewcommand{\thethm}{S\arabic{thm}}
\renewcommand{\thelem}{S\arabic{lem}}

\begin{center}
  \LARGE \textbf{Supplement to ``The trivariate wrapped Cauchy copula''} \\[1em]
  \normalsize
  Shogo Kato\textsuperscript{1}, Christophe Ley\textsuperscript{2}, Sophia Loizidou\textsuperscript{2}, Kanti V. Mardia\textsuperscript{3} \\
  \textsuperscript{1}Institute of Statistical Mathematics, Japan\\
  \textsuperscript{2}University of Luxembourg, Luxembourg \\
  \textsuperscript{3}University of Leeds, United Kingdom \\
\end{center}
\vspace{2em}
\begin{abstract}
    In this Supplementary Material, we provide additional information to the main paper ``A versatile trivariate wrapped Cauchy copula with applications to toroidal and cylindrical data''.
Section~\ref{sec:technical lemmas} provides the parameter space of the complex TWCC distribution,
Section~\ref{Sup_trigmom} contains results concerning the trigonometric moments of our TWCC density, Section~\ref{sec:interpretation} provides other properties of our proposed distribution and contour plots and Sections \ref{sec:correlation_coef} and \ref{sec:mom} consist of the correlation coefficients and method of moments estimation, respectively.
Section~\ref{sec:fisher_information} details the expression for the expected Fisher information matrix associated with the TWCC, while Section~\ref{sec:simulations} consists of Monte Carlo simulations to evaluate the performance of our MLE algorithm for our trivariate wrapped Cauchy copula. A multivariate extension of the TWCC density is investigated in Section~\ref{sec: possible_extensions}. Finally Section~\ref{sec:generalc1} discusses a generalization of our distribution \eqref{eq:tri_density}.
\end{abstract}

\section{Parameter space of the complex TWCC}\label{sec:technical lemmas}

We provide a technical lemma regarding the parameter space of the complex TWCC (\ref{eq:tri_density_c}) discussed in Section \ref{sec:Props} of the main paper.

\begin{lem}\label{lem_App}
For real-valued variables $\phi_1, \phi_2, \phi_3$, we have that $|\phi_i| > |\phi_j| + |\phi_{k}|$ for $(i,j,k)$ a certain permutation of $(1,2,3)$ if and only if
 $\phi_1 z_1 + \phi_2 z_2 + \phi_3 z_3 \neq 0$ for all $(z_1,z_2,z_3)'\in\Omega^3.$
\end{lem}

\begin{proof}
Let us start with the necessary condition. Without loss of generality, assume that $|\phi_1| > |\phi_2| + |\phi_{3}|$.
Then the (reverse) triangular inequality combined with straightforward calculations yields
\begin{eqnarray*}
|\phi_1 z_1 + \phi_2 z_2 + \phi_3 z_3|&\geq& ||\phi_1 z_1| - |\phi_2 z_2 + \phi_3 z_3||=||\phi_1| - |\phi_2 z_2 + \phi_3 z_3||\\
&=&|\phi_1| - |\phi_2 z_2 + \phi_3 z_3|
\geq|\phi_1| -( |\phi_2| + |\phi_3|)
>0.
\end{eqnarray*}
The sufficient condition requires some more steps. Without loss of generality assume that $\min\{\phi_1,\phi_2,\phi_3\}=\phi_1$. Then $z_1\neq-\frac{\phi_2z_2+\phi_3z_3}{\phi_1}$ and $|\phi_i/\phi_1|>1$ for $i=2,3$. From the former condition we can deduce that $\left|-\frac{\phi_2z_2+\phi_3z_3}{\phi_1}\right|<1$ or $\left|-\frac{\phi_2z_2+\phi_3z_3}{\phi_1}\right|>1$ for all $(z_2,z_3)'\in\Omega^2$. The special choices $z_2={\rm sgn}(\phi_2/\phi_1)\in\Omega$ and $z_3={\rm sgn}(\phi_3/\phi_1)\in\Omega$ lead to 
$$
\left|-\frac{\phi_2z_2+\phi_3z_3}{\phi_1}\right|=\left|\frac{\phi_2}{\phi_1}\right|+\left|\frac{\phi_3}{\phi_1}\right|>1
$$
by our second deduction above. Therefore $\left|-\frac{\phi_2z_2+\phi_3z_3}{\phi_1}\right|>1$ for all $(z_2,z_3)'\in\Omega^2$. Now choose again $z_2={\rm sgn}(\phi_2/\phi_1)\in\Omega$ but this time $z_3=-{\rm sgn}(\phi_3/\phi_1)\in\Omega$. From our established inequality we thus know that for these choices of $z_2, z_3$
$$
\left|-\frac{\phi_2z_2+\phi_3z_3}{\phi_1}\right|>1\Leftrightarrow\left|\left|\frac{\phi_2}{\phi_1}\right|-\left|\frac{\phi_3}{\phi_1}\right|\right|>1,
$$
and consequently either $|\phi_2|>|\phi_1|+|\phi_3|$ or $|\phi_3|>|\phi_1|+|\phi_2|$.
\end{proof}

\section{Trigonometric moments}\label{Sup_trigmom}
	For a trivariate circular random vector $(U_1,U_2,U_3)'$, its trigonometric moment {of order $(p_1,p_2,p_3)'$} is defined by
	$$
	\Phi (p_1,p_2,p_3) = E \left[ e^{{\rm i} (p_1 U_1 + p_2 U_2 + p_3 U_3 ) } \right],
	$$
	where $(p_1,p_2,p_3)' \in \mathbb{Z}^3$ is the order of the trigonometric moments.
	The following theorem shows that the trigonometric moments for the proposed distribution {TWCC($\boldsymbol{\rho}$)} can be expressed in simple form. 
	
	\begin{thm} \label{thm:moments}
		{Let $(U_1,U_2,U_3)'$ have the distribution {TWCC($\boldsymbol{\rho}$)} with $|\rho_{j k}| < |\rho_{i j} \rho_{i k}| / (|\rho_{ij}|+|\rho_{i k}|)$.
		Then {we have the following cases.}
  \begin{itemize} 
  \item[(i)] If $p_1+p_2+p_3 \neq 0$,
			$\Phi (p_1,p_2,p_3) = 0.$
		
		\item[(ii)] If $p_1+p_2+p_3 = 0$ and $p_i \geq 0$, then
		\begin{equation}
			\Phi (p_1,p_2,p_3)  = (-\rho_{j k})^{p_i} \sum_{n=0}^{p_i} {p_i \choose n} \rho_{i k}^{-n} \rho_{i j}^{-p_i+n} 
			\varphi_{j k}^{|p_j + n|},
			\label{eq:moments}
		\end{equation}
		where 
		\begin{equation}\label{theVARPHI}
		    \varphi_{j k}= \min \{ |\phi_{j k}| , |\phi_{j k}|^{-1} \} \phi_{j k} / |\phi_{j k}|
		\end{equation} and $\phi_{j k}$ is given by  \eqref{thephi}.
{
\item[(iii)] If $p_1+p_2+p_3=0$ and $p_i \geq 0$ and
\begin{itemize}
    \item[-] if $p_j \geq 0$, then the trigonometric moment simplifies to
$$
\Phi (p_1,p_2,p_3) = \varphi_{jk}^{p_j} \left\{ - \rho_{jk} \left( \frac{\varphi_{jk}}{\rho_{i k}} + \frac{1}{\rho_{ij}} \right) \right\}^{p_i};
$$
\item[-] if $p_j \leq -p_i$, then 
\begin{equation}
\Phi (p_1,p_2,p_3) = \varphi_{jk}^{-p_j} \left\{ - \rho_{jk} \left( \frac{1}{\varphi_{jk} \rho_{ik}} + \frac{1}{\rho_{ij}} \right) \right\}^{p_i}.  \label{eq:tm_remark}
\end{equation}
\end{itemize}
  
		\item[(iv)] If $p_1+p_2+p_3 = 0$ and $p_i <0 $, then
		$ \Phi (p_1,p_2,p_3) = \Phi (-p_1, -p_2, -p_3).$
  }
 
 
\item[(v)] If $p_1+p_2+p_3=0$ and $p_i=0$, the trigonometric moment is given by 
$$\Phi (p_1,p_2,p_3) = \varphi_{j k}^{|p_j|}.$$
\end{itemize}}
	\end{thm}

\begin{proof}
Without loss of generality, assume that $(i,j,k)=(1,2,3)$, namely, $|\rho_{23}| < |\rho_{12} \rho_{13}|/(|\rho_{12}|+|\rho_{13}|).$

For convenience, transform the random vector $(U_1,U_2,U_3)'$ into complex form $(Z_1,Z_2,Z_3)' = (e^{{\rm i}U_1},e^{{\rm i} U_2}, e^{{\rm i} U_3})'$ which has the density (\ref{eq:tri_density_c}).
It follows from Theorems \ref{thm:marginals}(i) and \ref{thm:conditionals}(iii) that the density of $(Z_1,Z_2,Z_3)'$ can be expressed as
$$
tc(\zb ; \rhob) = \frac{1}{(2\pi)^3} \frac{|1-|\phi_{23} |^2|}{|z_2 \overline{z}_3 - \phi_{23} |^2} \frac{ \left| 1-|( \phi_2 z_2 + \phi_3 z_3)/\phi_1|^2 \right| }{|z_1 + ( \phi_2 z_2 + \phi_3 z_3)/\phi_1|^2}.
$$
Then the trigonometric moments can be calculated as
\begin{align}
	\Phi (p_1,p_2,p_3) = \ & \int_{\Omega^2} \frac{z_2^{p_2} z_3^{p_3}}{(2\pi)^2} \frac{ \left| 1-|\phi_{23}|^2 \right|}{|z_2 \overline{z}_3 - \phi_{23}|^2} \int_{\Omega} z_1^{p_1} \frac{1}{2\pi} \frac{ \left| 1-|( \phi_2 z_2 + \phi_3 z_3)/\phi_1|^2 \right| }{|z_1 + ( \phi_2 z_2 + \phi_3 z_3)/\phi_1|^2} dz_1 dz_2 dz_3 \nonumber \\
	= \ & \int_{\Omega^2} \frac{z_2^{p_2} z_3^{p_3}}{(2\pi)^2} \frac{\left| 1-|\phi_{23}|^2 \right|}{|z_2 \overline{z}_3 - \phi_{23}|^2} \left( - \frac{\phi_2 z_2 + \phi_3 z_3}{\phi_1} \right)^{p_1} dz_2 dz_3 \nonumber \\
	\begin{split}  \label{eq:tm}
	= \ & \frac{1}{(2\pi)^2} \sum_{n=0}^{p_1} {p_1 \choose n} \left( - \frac{\phi_2}{\phi_1} \right)^n \left( - \frac{\phi_3}{\phi_1} \right)^{p_1-n} \\
	& \ \times \int_{\Omega} z_2^{p_2+n} z_3^{p_3+p_1-n} \frac{\left| 1-|\phi_{23}|^2 \right|}{|z_2 \overline{z}_3 - \phi_{23}|^2} dz_2 dz_3.
	\end{split}   
\end{align}
The second equality follows from Section 1.4 of \cite{suppMCC96} and
$$
\left| \frac{-(\phi_2 z_2 + \phi_3 z_3)}{\phi_1} \right| \leq \frac{|\phi_2|+|\phi_3|}{|\phi_1|}  = |\rho_{23}| \frac{|\rho_{12}| + |\rho_{13}|}{|\rho_{12} \rho_{13}|} < 1.
$$
The integration in (\ref{eq:tm}) can be calculated using equation (4.3) of \cite{suppK09} and equation (\ref{eq:cauchy_equiv}) as
$$
	\frac{1}{(2\pi)^2}\int_{\Omega} z_2^{p_2+n} z_3^{p_3+p_1-n} \frac{\left| 1-|\phi_{23}|^2 \right|}{|z_2 \overline{z}_3 - \phi_{23}|^2} dz_2 dz_3 = \left\{
	\begin{array}{ll}
	\varphi_{23}^{|p_2+n|}, & p_1 + p_2 + p_3 =0 ,\\
	0, & p_1+p_2+p_3 \neq 0,
	\end{array}
	\right.
$$
where $\varphi_{23}=\min \{ |\phi_{23}|,|\phi_{23}|^{-1} \} \phi_{23}/|\phi_{23}|$.
Then it follows that $\Phi (p_1,p_2,p_3) = 0 $ if $p_1+p_2+p_3 \neq 0$ and therefore (i) is proved.
If $p_1+p_2+p_3 =0$, we have
\begin{align*}
	\Phi (p_1,p_2,p_3) & = \sum_{n=0}^{p_1} {p_1 \choose n} \left( - \frac{\phi_2}{\phi_1} \right)^n \left( -\frac{\phi_3}{\phi_1} \right)^{p_1-n} \varphi_{23}^{|p_2+n|}  \\
	 & = (-\rho_{23})^{p_1} \sum_{n=0}^{p_1} {p_1 \choose n} \rho_{13}^{-n} \rho_{12}^{-p_1+n} 
	 \varphi_{23}^{|p_2 + n|}
\end{align*}
as required in (ii).
The second equality follows from the equation $-\phi_i/\phi_1 = - \mbox{sgn}(\rho_{1 j})|\rho_{23}|$ $/ \{ \mbox{sgn}(\rho_{23}) |\rho_{1 j}| \} = - \rho_{23}/\rho_{1j} $, where $i,j=2,3$, $i \neq j$.
{
In particular, if the additional assumption $p_2\geq 0$ holds, the binomial theorem implies
\begin{align*}
\lefteqn{\Phi (p_1,p_2,p_3) } \hspace{0.7cm} \\
& = \sum_{n=0}^{p_1} {p_1 \choose n} \left( - \frac{\phi_2}{\phi_1} \right)^n \left( -\frac{\phi_3}{\phi_1} \right)^{p_1-n} \varphi_{23}^{p_2+n} \\
& = \varphi_{23}^{p_2} (-\rho_{23})^{p_1} \sum_{n=0}^{p_1} {p_1 \choose n} \left( \frac{\varphi_{23}}{\rho_{13}}\right)^n \left( \frac{1}{\rho_{12}} \right)^{p_1-n}  =  \varphi_{23}^{p_2} \left\{ - \rho_{23} \left( \frac{\varphi_{23}}{\rho_{13}} + \frac{1}{\rho_{12}} \right) \right\}^{p_1}.
\end{align*}
Similarly, for $p_2 \leq -p_1$, we have 
\begin{align*}
\Phi (p_1,p_2,p_3) & = \sum_{n=0}^{p_1} {p_1 \choose n} \left( - \frac{\phi_2}{\phi_1} \right)^n \left( -\frac{\phi_3}{\phi_1} \right)^{p_1-n} \varphi_{23}^{-(p_2+n)} \\
& = \varphi_{23}^{-p_2} \left\{ - \rho_{23} \left( \frac{1}{\varphi_{23} \rho_{13}} + \frac{1}{\rho_{12}} \right) \right\}^{p_1},
\end{align*}
which gives (iii).
Item (iv) of the theorem holds because 
\begin{align*}
\Phi (-p_1,-p_2,-p_3) & = E [ e^{-{\rm i}(p_1 U_1 + p_2 U_2 + p_3 U_3)}] = \overline{E [ e^{{\rm i}(p_1 U_1 + p_2 U_2 + p_3 U_3)} ]} \\
& = E [ e^{{\rm i}(p_1 U_1 + p_2 U_2 + p_3 U_3)} ] = \Phi(p_1,p_2,p_3).
\end{align*}
Finally, it is straightforward to show (v) by substituting $p_1=0$ in (ii).
}
\end{proof}

\cite{suppKP15} obtained the expression for the trigonometric moments of the bivariate wrapped Cauchy distribution using the residue theorem.
Following the approach given in the proof of Theorem 2 of their paper, it is also possible to express the trigonometric moments in the setting (ii) using the residues.
However, for our model, the expression (\ref{eq:moments}) seems more practical because this expression does not involve the calculation of derivatives which appear in the residues.

{We} note that Theorem \ref{thm:moments}(i) is not specific to the {TWCC} distribution and can be generalized to a more general multivariate family.
For details, see Theorem \ref{thm:multivariate}(ii).
 We will come back to these results when we deal with the parameter estimation by the method of moments in Section~\ref{sec:mom}.

\section{Other properties of the TWCC density and contour plots} \label{sec:interpretation}
{
Here we present some other properties of the {TWCC} density including the limiting cases, display the contour plots, and discuss the interpretation of the parameters.
{
We first consider limiting cases of the parameter values of {TWCC($\boldsymbol{\rho}$)} with the  constraint on  $\rho_{23}$, meaning
\begin{equation}\label{positivity}
|\rho_{23}| < \frac{|\rho_{12} \rho_{13}|}{ |\rho_{12}| + |\rho_{13}|}.
\end{equation}
Note that this order of the indices is selected without any loss of generality and can be permuted.
Under this assumption, the two cases of the parameters, namely, with and without the constraint $\rho_{12}\rho_{13}\rho_{23}=1$, are discussed.

\begin{prop} \label{prop:limiting_without}
Suppose that the positivity constraint (\ref{positivity}) holds and that the constraint $\rho_{12} \rho_{13} \rho_{23} =1$ does not hold in general.
Then the following limiting results hold for {TWCC($\boldsymbol{\rho}$)}:
\begin{enumerate}
\item[(i)] As $\rho_{23} \rightarrow 0$, {TWCC($\rhob$)} converges to the uniform distribution on the torus $[0,2\pi)^3$.
\item[(ii)] As $|\rho_{12}| \rightarrow \infty$, {TWCC($\rhob$)} converges to the distribution with density
\begin{equation}
t(\ub; \rhob) = \frac{1}{(2\pi)^3} \frac{1-\rho_{23.13}^2}{1 + \rho_{23.13}^2 + 2 \rho_{23.13} \cos (u_1 - u_2) }, \quad 0 \leq u_1,u_2 \leq 2\pi, \label{eq:limit}
\end{equation}
where 
{$\rho_{23.13} = \rho_{23}/\rho_{13}$ satisfies $|\rho_{23.13}| < 1$}.
Therefore the bivariate vector $(U_1, U_2)'$ of this limiting distribution is independent of $U_3$ (which is uniformly distributed) and follows the bivariate wrapped Cauchy copula (\ref{eq:wc_wj}) labelled BWC(0,$\rho_{23.13}$).
\end{enumerate}
\end{prop}

\begin{proof}
Multiplying $\rho_{23}/(\rho_{12} \rho_{13})$ to both the numerator and denominator of the {TWCC} density (\ref{eq:tri_density}), we obtain
\begin{align}
\lefteqn{ t(\ub; \rhob) } \hspace{0.6cm} \nonumber \\
= \ &  \frac{1}{(2\pi)^3} \left[ 1 + \rho_{23.13}^4 + \left( \frac{\rho_{23}}{\rho_{12}} \right)^4 - 2 \rho_{23.13}^2 - 2 \left( \frac{\rho_{23}}{\rho_{12}} \right)^2 - 2 \left( \rho_{23.13} \cdot\frac{\rho_{23}}{\rho_{12}} \right)^2 \right]^{1/2} \nonumber \\
 & \Bigg/ \Biggl[ 1 + \rho_{23.13}^2 + \left( \frac{\rho_{23}}{\rho_{12}}\right)^2 + 2 \biggl\{ \rho_{23.13} \cos (u_1-u_2) + \frac{\rho_{23}}{\rho_{12}} \cos (u_1-u_3) \nonumber \\
 & \quad \quad + \rho_{23.13} \cdot \frac{\rho_{23}}{\rho_{12}} \cos (u_2-u_3) \biggr\} \Biggr]. \label{eq:limiting_rho}
\end{align}
From this expression, it is straightforward to see that
$ t(\ub; \rhob) \rightarrow 1/(2\pi)^3 $ as $\rho_{23} \rightarrow 0$, proving (i), and that
$$
t(\ub; \rhob) \rightarrow \frac{1}{(2\pi)^3} \frac{1-\rho_{23.13}^2}{1 + \rho_{23.13}^2 + 2 \rho_{23.13} \cos (u_1 - u_2) } \quad \mbox{as} \quad |\rho_{12}| \rightarrow \infty,
$$
proving (\ref{eq:limit}).
Since the parameter constraint (\ref{positivity}) implies $|\rho_{23}| < |\rho_{13}|$, we have $|\rho_{23.13}| <1$.
Finally, it follows from the form of the limiting density (\ref{eq:limit}) that $U_3$ is independent of $(U_1,U_2)'$ and follows the uniform distribution on the circle and that the marginal of $(U_1,U_2)'$ is the bivariate wrapped Cauchy copula BWC(0,$\rho_{23.13}$).
\end{proof}


{
Proposition \ref{prop:limiting_without}(ii) implies that $\rho_{23.13}$, which is referred to as a partial dependence parameter, can be interpreted as the parameter that controls the strength of dependence between $U_1$ and $U_2$ when $|\rho_{12}|$ is large.
To be more specific, if $|\rho_{23.13}| \simeq 1$, then $U_1$ and $U_2$ have a strong dependence.
In particular, for any $\varepsilon >0$ and a fixed sign of $\rho_{23.13}$, it holds that $ \lim_{|\rho_{13}| \rightarrow |\rho_{23}| +0} P(|U_i-U_j - \pi I(\rho_{23.13} >0)| < \varepsilon) \rightarrow 1$.
On the other hand, if $|\rho_{23.13}| \ll 1 $, then the dependence between $U_1$ and $U_2$ becomes weak.
The positive and negative values of $\rho_{23.13}$ imply the modes at $u_1=u_2+\pi$ and at $u_1=u_2$, respectively; see \cite{suppK09}. 

Let us define 
$$
T_1=U_1-U_2, \quad T_2=U_2, \quad T_3=U_3,
$$
then it can be seen that $T_1$  is univariate wrapped Cauchy which is independent of   $T_2$ and $T_3$ which are independently uniformly distributed. Now as $\rho_{23.13} \rightarrow 1$, $T_1$  tends to univariate Cauchy on the real line.
This can be seen from the fact that with $t_1= \varepsilon v/ \{ 2 (1-\varepsilon)^{1/2} \}$, $\rho_{23.13}=1-\varepsilon$ for small $\varepsilon >0$,
\begin{align*}
f(t_1) & = \frac{1}{2\pi} \frac{1-\rho_{23.13}^2}{1+\rho_{23.13}^2-2 \rho_{23.13} \cos (t_1)} \simeq \frac{1}{2\pi} \frac{1-\rho_{23.13}^2}{ (1-\rho_{23.13})^2 +2 \rho_{23.13} \{ \varepsilon^2 v^2/ 4 (1-\varepsilon) \} /2} \\
 & = \frac{1}{\pi} \frac{1}{1 +v^2}. 
\end{align*}


{
Next, with the additional constraint $\rho_{12}\rho_{13}\rho_{23}=1$, the limiting cases of our model are more involved.

\begin{prop} \label{prop:limiting_with}
Let $\{ \rhob_n \}_{n=1}^{\infty} $ be a sequence of parameters of the {TWCC} density.
Assume that $ \rhob_n = (\rho_{12,n},\rho_{13,n},\rho_{23,n})'$ satisfies  $\rho_{12,n} \rho_{13,n} \rho_{23,n}=1,\ \rho_{12,n}=O(n^\iota),\ \rho_{13,n}=O(n^\kappa),$ $\rho_{23,n}=O(n^\lambda ),$ and $\iota \geq \kappa$.
Then the following hold for {TWCC($\rhob_n$)}:
\begin{enumerate}
    \item[(i)] If $\kappa > \lambda$, then {TWCC($\rhob_n$)} converges to the uniform distribution on the torus as $n \rightarrow \infty$.
\item[(ii)] Let $\rho_{12,n}=a^{-1} n, \ \rho_{13,n}=n^{-1/2}$ and $\rho_{23,n} = a n^{-1/2}$ $(0<a<(-1+\sqrt{5})/2) $.
Then $\kappa=\lambda$ holds and {TWCC($\rhob_n$)} converges to the distribution (\ref{eq:limit}) with $\rho_{23.13} = a$ as $n \rightarrow \infty$.
\end{enumerate}
\end{prop}

\begin{proof}

An alternative expression for the density of {TWCC($\rhob_n$)} can be obtained by substituting $\rhob=\rhob_n$ in (\ref{eq:limiting_rho}).
This expression, along with $\rho_{23,n}/\rho_{13,n} =O(n^{\lambda-\kappa})$ and $\rho_{23,n}/\rho_{12,n}=O(n^{\lambda-\iota}) $, implies that, under the assumption in (i), namely, $\kappa > \lambda$, $t(\ub;\rhob) \rightarrow 1/(2\pi)^3$ as $n \rightarrow \infty$.
If the assumptions in (ii) holds, then we have $\rho_{23,n}/\rho_{13,n} = a$ and $\rho_{23,n}/\rho_{12,n}= a^2 n^{-3/2}$, and hence
$$
t(\ub;\rhob) \rightarrow \frac{1}{(2 \pi)^3} \frac{1-a^2}{1+a^2+2 a \cos (u_1-u_2)} \quad \mbox{as} \quad n \rightarrow \infty.
$$
This limiting distribution is the distribution (\ref{eq:limit}) with $\rho_{23.13}=a$.

Note that for $a>0$, the constraint $a<(-1+\sqrt{5})/2$ is necessary to ensure the inequality (\ref{positivity}) for any $n$.
This can be seen by first noticing that the inequality (\ref{positivity}) implies 
$$
\frac{1}{1+a n^{-3/2}} \cdot \frac{1}{a} \geq 1.
$$
Since the left-hand side of this inequality is minimized for $n=1$, it follows that $ 1 / \{(1+a)a \} \geq 1$.
Thus we have $a<(-1+\sqrt{5})/2$.
\end{proof}

Proposition \ref{prop:limiting_with}(ii) implies that if $\kappa=\lambda$, the convergence of {TWCC($\rhob_n$)} as $n \rightarrow \infty$ differs depending on the specific forms of $\rhob_n$.


}
}

Now we present some more properties of the {TWCC} density.
{Their proofs are straightforward and omitted.}

\begin{prop} \label{prop:others}
Let $t(\ub ; \rhob)$ be the density (\ref{eq:tri_density}).
Then it can be shown that it has the following properties:
\begin{enumerate}
 \item[(i)] $
        \!
        \begin{aligned}[t]
        t(u_1,u_2,u_3; \rho_{12}, \rho_{13},\rho_{23}) & = t(u_1+\pi, u_2 ,u_3; -\rho_{12},-\rho_{13},\rho_{23}) \\
		& = t(u_1, u_2+\pi, u_3; -\rho_{12},\rho_{13},-\rho_{23}) \\
		& = t(u_1, u_2, u_3+\pi; \rho_{12},-\rho_{13},-\rho_{23}),
        \end{aligned}
        $
	\item[(ii)] $t(\ub; \rhob) = -t(\ub; -\rhob) $,
	\item[(iii)] {for $\ub \in[-2\pi,2\pi)^3$,} $t(-\ub; \rhob) = t(\ub; \rhob)$.
\end{enumerate}
\end{prop}

The property (i) of Proposition~\ref{prop:others} implies that changing the signs of two parameters corresponds to the location shift of a variable. Property (ii) shows how crucial the condition $\rho_{12}\rho_{13}\rho_{23}>0$ (or \eqref{identcond}) actually is as it prevents from negative densities for other parameter combinations
Finally property (iii) implies that our proposed density is symmetric about its center. 


{

\begin{figure}[p]
	\begin{tabular}{cccc}
		\centering 
		\includegraphics[trim=0.7cm 1cm 0cm 0cm,clip,width=3cm,height=3cm]{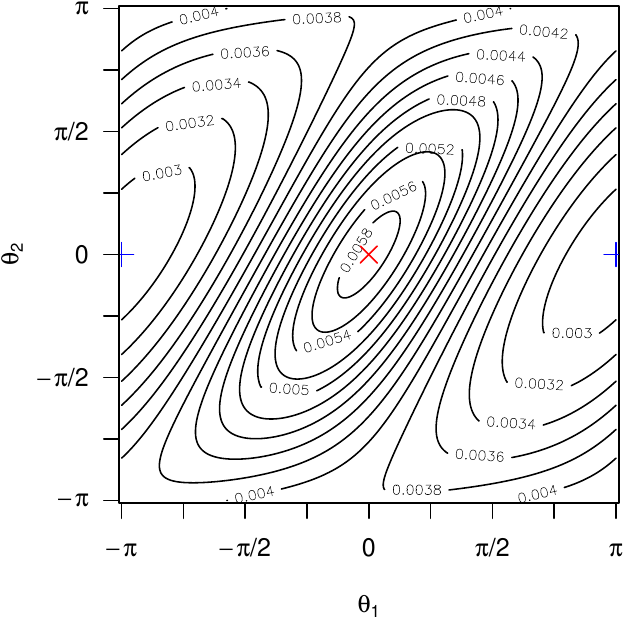} &
   \includegraphics[trim=0.7cm 1cm 0cm 0cm,clip,width=3cm,height=3cm]{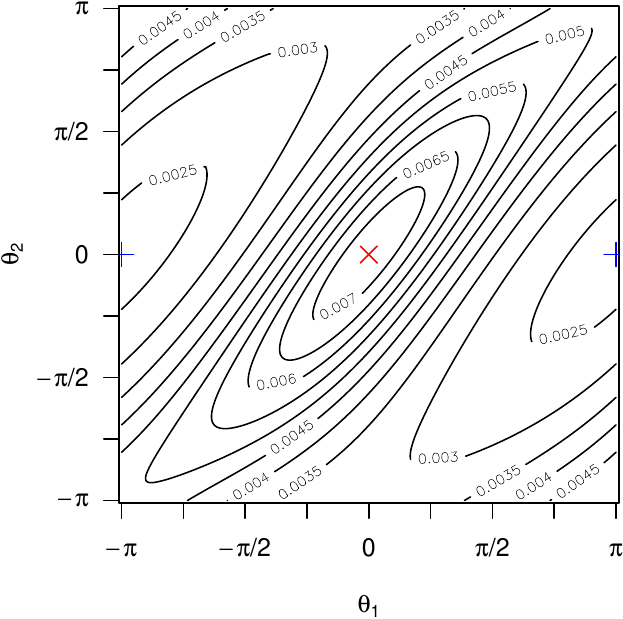} &
		\includegraphics[trim=0.7cm 1cm 0cm 0cm,clip,width=3cm,height=3cm]{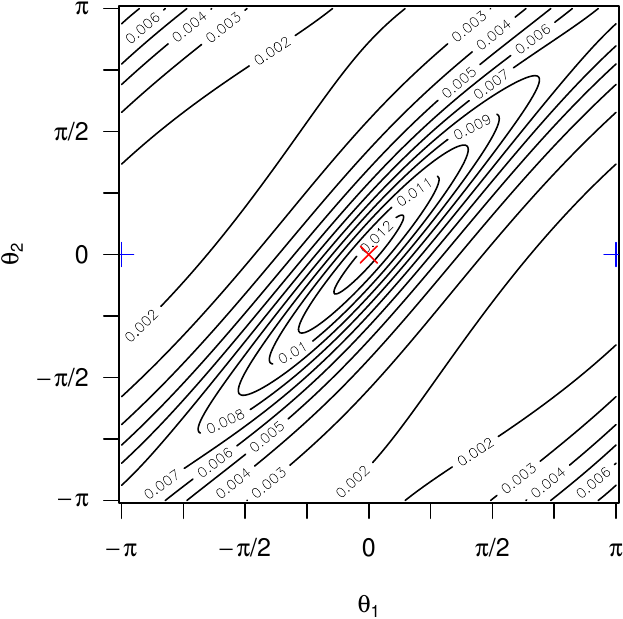} &
		\includegraphics[trim=0.7cm 1cm 0cm 0cm,clip,width=3cm,height=3cm]{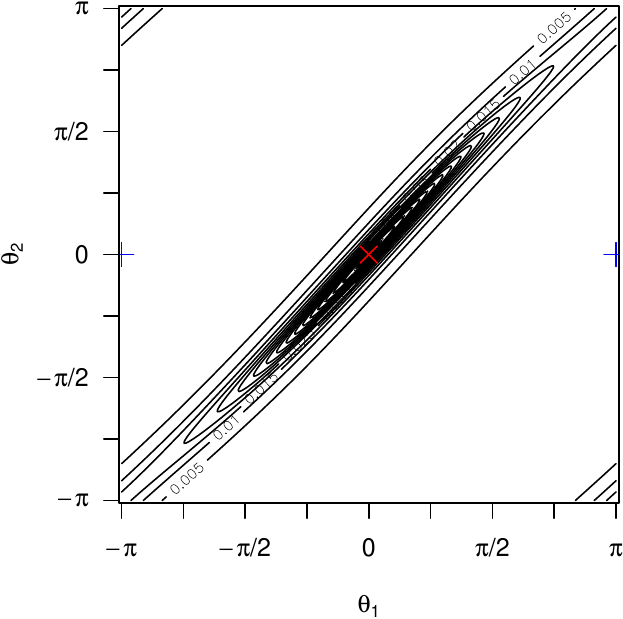}\\
		\hspace{0.3cm} (a) & \hspace{0.45cm}(b) & \hspace{0.45cm}(c) & \hspace{0.55cm}(d) \vspace{0.2cm}\\
			&
	\includegraphics[trim=0.7cm 1cm 0cm 0cm,clip,width=3cm,height=3cm]{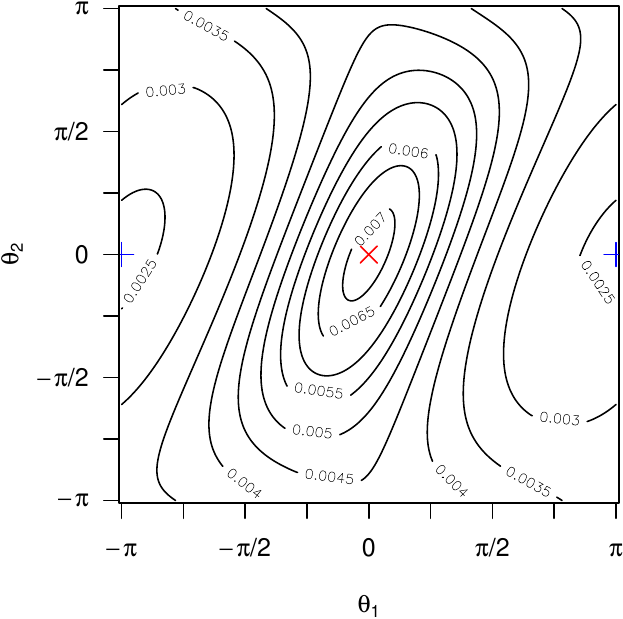} &
		\includegraphics[trim=0.7cm 1cm 0cm 0cm,clip,width=3cm,height=3cm]{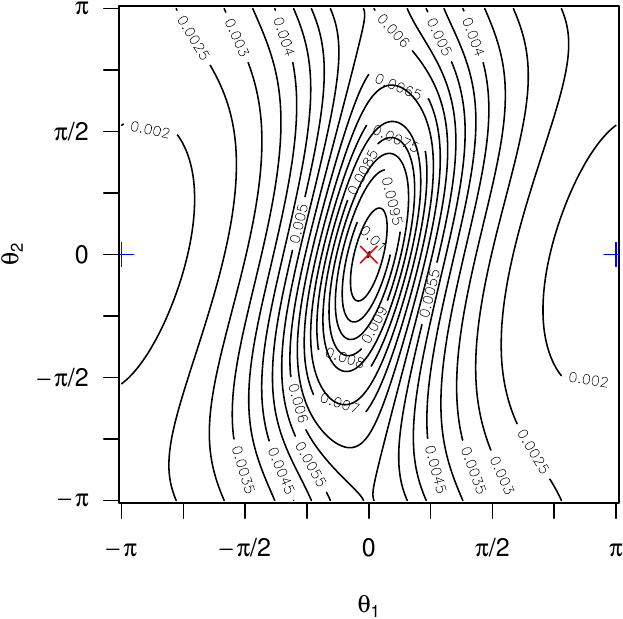} &
		\includegraphics[trim=0.7cm 1cm 0cm 0cm,clip,width=3cm,height=3cm]{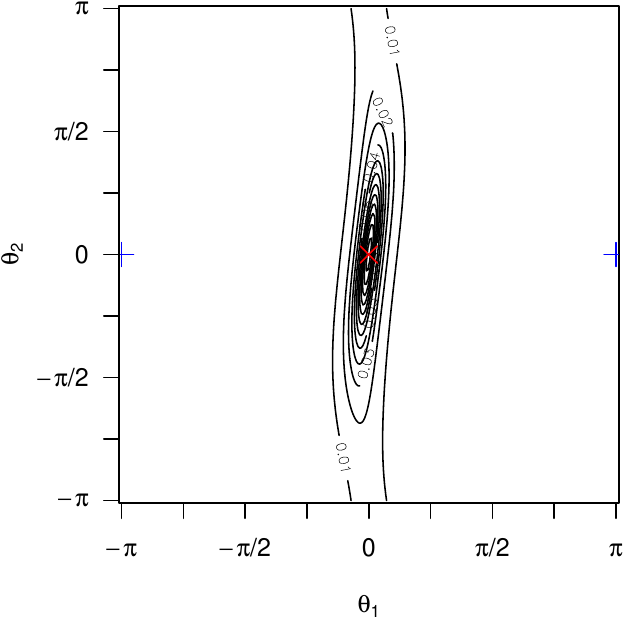} \\	
		\hspace{0.5cm} & \hspace{0.45cm}(e) & \hspace{0.45cm}(f) & \hspace{0.55cm}(g) \vspace{0.2cm}\\
	
	&
	\includegraphics[trim=0.7cm 1cm 0cm 0cm,clip,width=3cm,height=3cm]{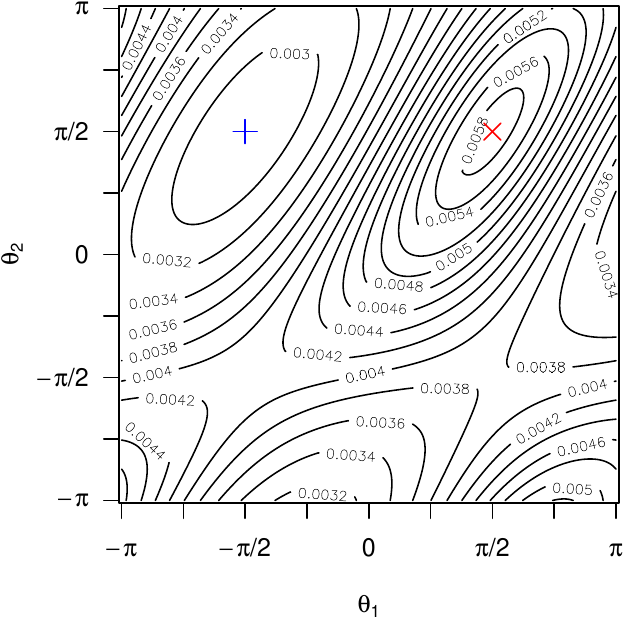} &
	\includegraphics[trim=0.7cm 1cm 0cm 0cm,clip,width=3cm,height=3cm]{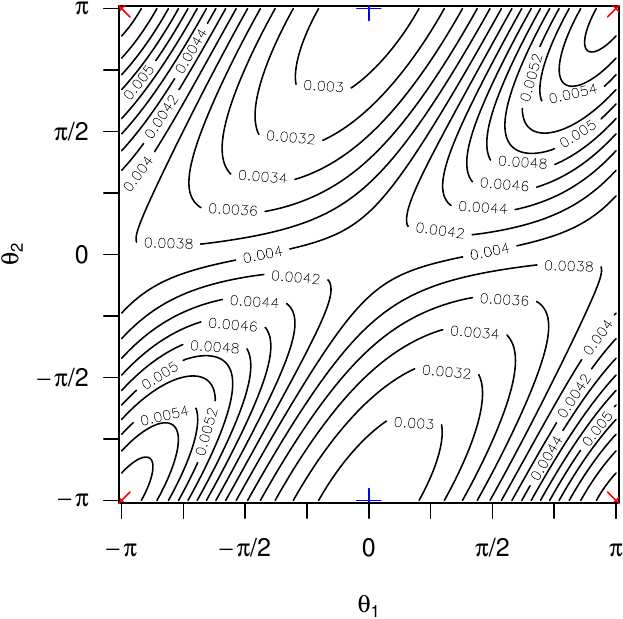} &
	\includegraphics[trim=0.7cm 1cm 0cm 0cm,clip,width=3cm,height=3cm]{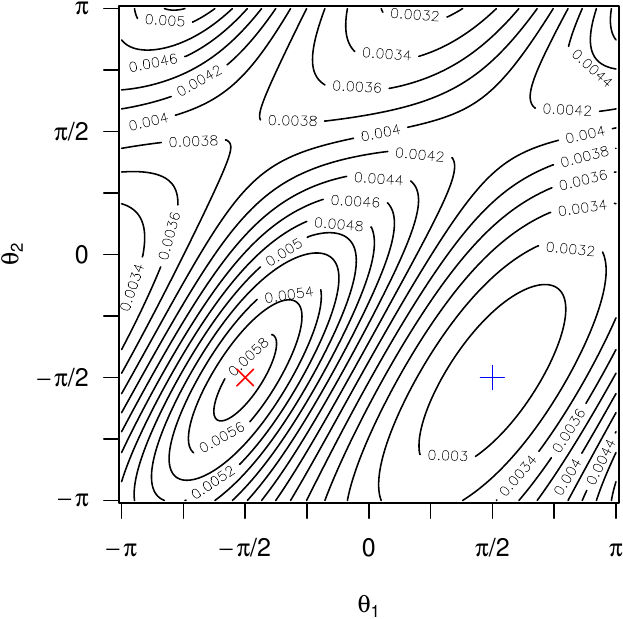} \\
	\hspace{0.35cm}  & \hspace{0.45cm}(h) & \hspace{0.45cm}(i) & \hspace{0.55cm}(j) \vspace{0.2cm}\\
\end{tabular}
\caption[]{{Contour plots of 
	the density of {TWCC($\boldsymbol{\rho}$)} with $\rho_{23}=0.2$ and: (a)--(d) $(-\rho_{12},u_3)=(2.5,0)$ and (a) $-\rho_{13}=2$, (b) $-\rho_{13}=1$ and (c) $-\rho_{13}=0.5$ and (d) $-\rho_{13}=0.25$, (e)--(g) $(-\rho_{13},u_3)=(2,0)$ and (e) $-\rho_{12}=1$, (f) $-\rho_{12}=0.5$ and (g) $-\rho_{12}=0.25$, and (h)--(j) $(-\rho_{12},-\rho_{13})=(2.5,2)$ and (h) $u_3=\pi/2$, (i) $u_3=\pi$ and (j) $u_3=3 \pi/2$.
	The $x$-axis represents the value of $u_1$, while the $y$-axis denotes the value of $u_2$.
 The symbols `$\times$' (red) and `$+$' (blue) denote the modes and antimodes of density (\ref{eq:tri_density}), respectively.\label{fig:densities} }}
\end{figure}

Figure \ref{fig:densities} shows the density of {TWCC($\boldsymbol{\rho}$)} for fixed values of the third component $u_3$
  and selected values of the parameters subject to condition (\ref{positivity}), but not generally subject to $\rho_{12} \rho_{13} \rho_{23} = 1$. 
{
The figure shows that the {TWCC($\rhob$)} density can have a wide range of dependence structure and strength of dependence.}
The plots (a)--(d) indicate that the dependence between $U_1$ and $U_2$ becomes strong when the parameter $|\rho_{13}|$ is close to $|\rho_{23}|$ or equivalently $|\rho_{23.13}|$ approaches 1. 
The comparison among the plots (a) and (e)--(g) suggests that the closer the value of $ |\rho_{12}| $ is to its boundary, namely, $ |\rho_{12}| \simeq |\rho_{13} \rho_{23}| / (|\rho_{13}| - |\rho_{23}|) \simeq 0.23 $, the greater the concentration of $ U_1$  for fixed values of $ u_3$.
Also, it can be seen from plots (a) and (h)--(j) that the values of $u_3$ control the location of $(u_1, u_2)'$ when the density {TWCC($\boldsymbol{\rho}$)} given by (\ref{eq:tri_density}) is viewed as a function of $(u_1, u_2)'$.
Finally, as visual confirmation of Theorem \ref{thm:modes}, the modes and antimodes of {TWCC($\boldsymbol{\rho}$)} in all the frames are at $u_1 = u_2 = u_3$ and $u_1 + \pi = u_2 = u_3$, respectively.
}

For a better interpretation of the parameters of the {TWCC} with the constraint $\rho_{12} \rho_{13} \rho_{23}=1$, it is advantageous to reparametrize the parameters in terms of reparametrized versions of $\rho_{12}$ and $\rho_{23.13}$ given by
\begin{equation}\label{fulldep}
\rho_{12}^* = \left\{ \frac{ 1-\left|\rho_{23.13}^*\right| }{\left|\rho_{23.13}^*\right|^{1/2}} \right\}^{2/3} \rho_{12}, \quad \rho_{23.13}^* = \frac{\rho_{23}}{|\rho_{13}|},
\end{equation}
respectively.
The ranges of these reparametrized parameters are $\left|\rho_{12}^*\right| > 1$ and $|\rho_{23.13}^*| > 0$.

Here $\rho_{12}^*$ is called the full dependence parameter.
Unlike $\rho_{12}$ itself, $\rho_{12}^*$ has a simpler parameter range and can be directly interpreted without considering the value of $\rho_{23.13}$ or $\rho_{23.13}^*$.
The parameter $\rho_{12}^*$ can be derived first by noting that $\rho_{23}= \rho_{23.13} \cdot \rho_{13} $ and the positivity constraint (\ref{positivity}) implies $|\rho_{13}| < \{ (1-|\rho_{23.13}|) / |\rho_{23.13}| \} |\rho_{12}| $.
In addition it follows from $\rho_{12} \rho_{13} \rho_{23}=1$ that $|\rho_{13}|=1/(\rho_{12} \cdot \rho_{23.13})^{1/2}$.
Summarizing these results, the range of $\rho_{12}$ can be expressed as $|\rho_{12}| > \{ |\rho_{23.13}|^{1/2} / (1-|\rho_{23.13}|) \}^{2/3} $.
Then, in order to achieve a simpler parameter range and maintain its original sign, $\rho_{12}$ can be reparameterized as (\ref{fulldep}).

The parameter $\rho_{23.13}$ is reparametrized as $\rho_{23.13}^*$ because the original signs of $\rho_{13}$ and $\rho_{23}$ are not uniquely determined by the sign of $\rho_{23.13}$ and the other parameter conditions.
For example, $\rho_{23.13} > 0$ implies that both $\rho_{13}, \rho_{23} > 0$ and $\rho_{13}, \rho_{23} < 0$ are possible, and the other parameter conditions do not eliminate either of these possibilities.
On the other hand, with the reparametrized parameters $\rho_{12}^*$ and $\rho_{23.13}^*$, there is a one-to-one correspondence between the parameter space of $(\rho_{12}, \rho_{13}, \rho_{23})$ and that of $(\rho_{12}^*, \rho_{23.13}^*)$.

\begin{figure}[p]
	\centering {\footnotesize
		\hspace{1.6cm} $\rho_{23.13}^*=0.1$ \hspace{1.3cm} $\rho_{23.13}^*=0.5$ \hspace{1.3cm} $\rho_{23.13}^*=0.9$ \vspace{0.2cm}\\
		\raisebox{1.7cm}{$\rho_{12}^*=-5$} 
		\includegraphics[trim=0.7cm 1cm 0cm 0cm,clip,width=3.1cm,height=3.1cm]{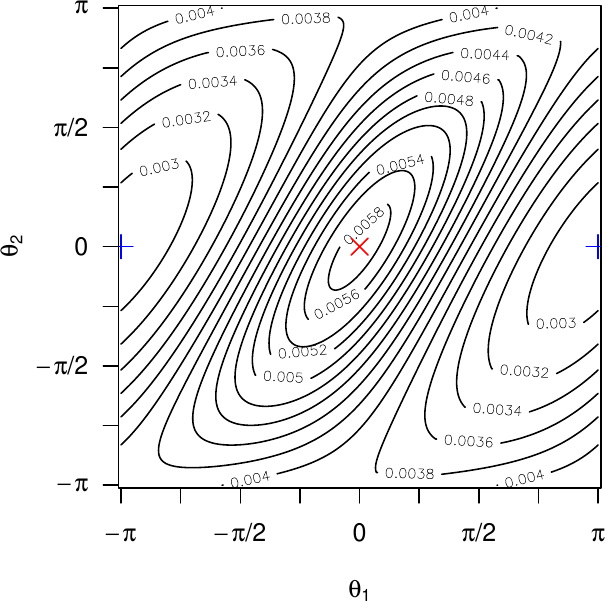}
		\includegraphics[trim=0.7cm 1cm 0cm 0cm,clip,width=3.1cm,height=3.1cm]{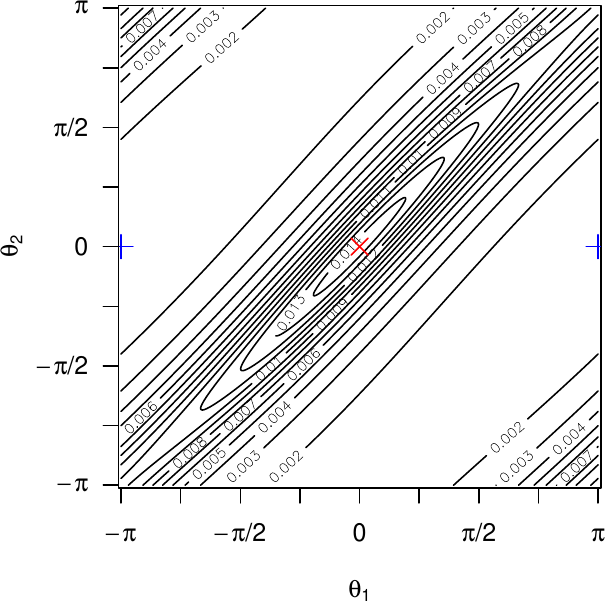}
		\includegraphics[trim=0.7cm 1cm 0cm 0cm,clip,width=3.1cm,height=3.1cm]{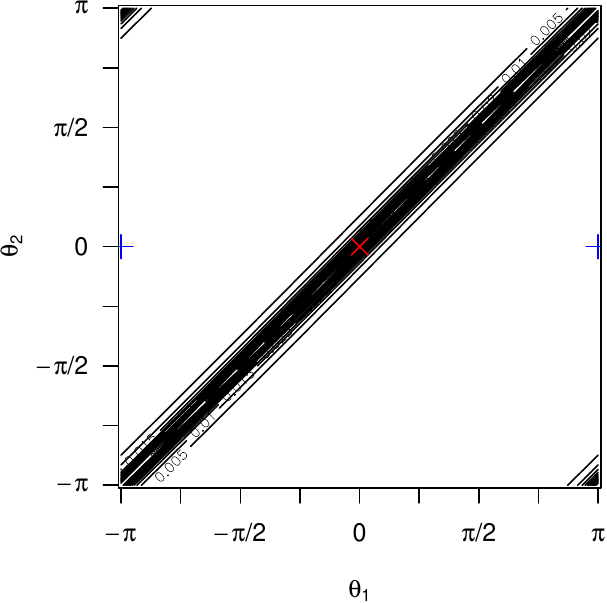} \vspace{0.2cm}\\
		\raisebox{1.7cm}{$\rho_{12}^*=-2$} 
		\includegraphics[trim=0.7cm 1cm 0cm 0cm,clip,width=3.1cm,height=3.1cm]{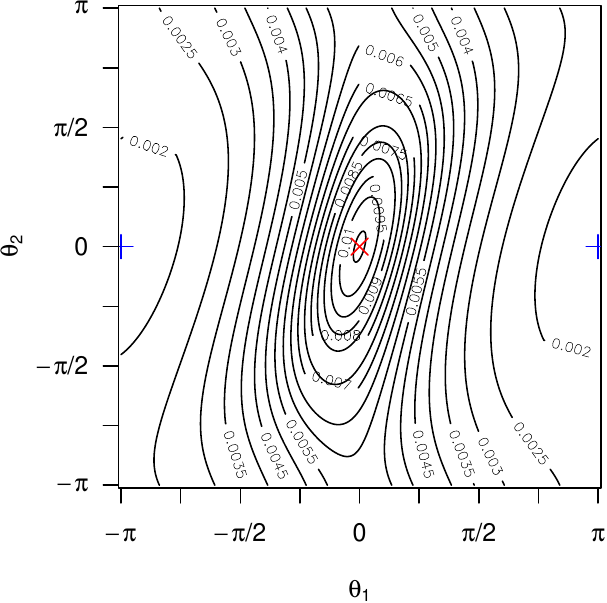}
		\includegraphics[trim=0.7cm 1cm 0cm 0cm,clip,width=3.1cm,height=3.1cm]{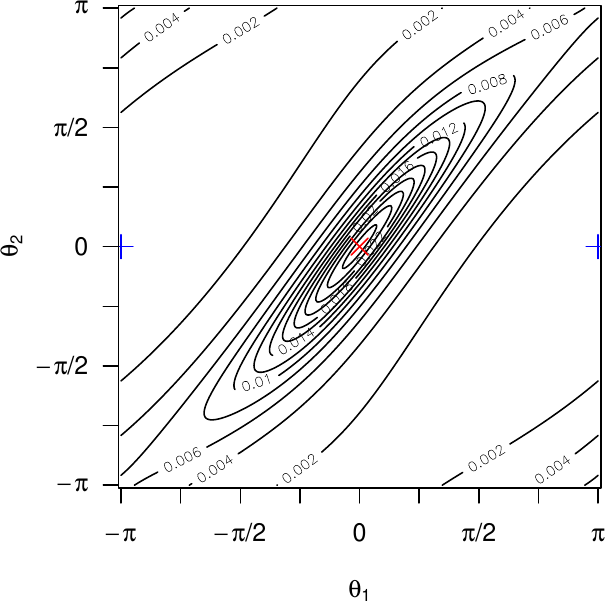}
		\includegraphics[trim=0.7cm 1cm 0cm 0cm,clip,width=3.1cm,height=3.1cm]{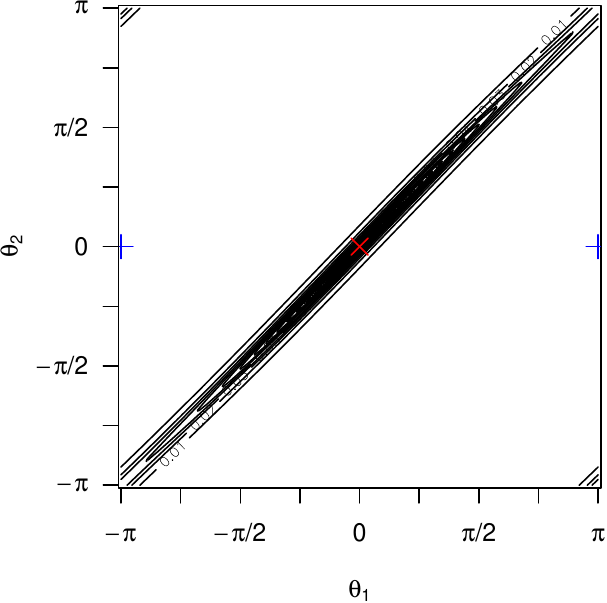} \vspace{0.2cm}\\
		\raisebox{1.7cm}{$\rho_{12}^*=-1.1$} 
		\includegraphics[trim=0.7cm 1cm 0cm 0cm,clip,width=3.1cm,height=3.1cm]{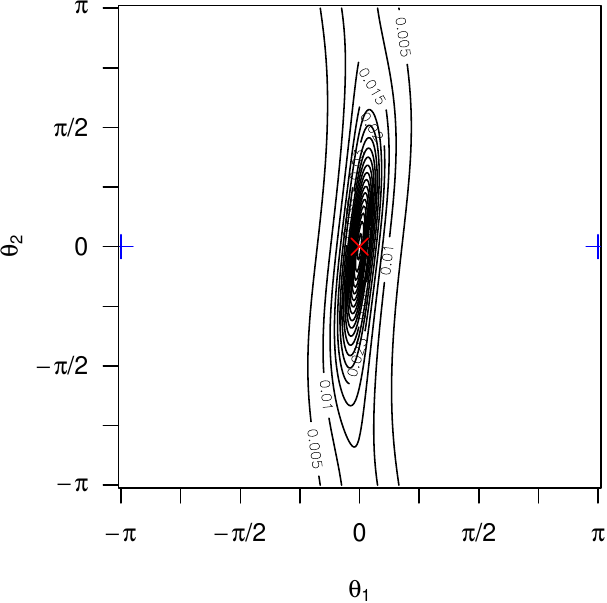}
		\includegraphics[trim=0.7cm 1cm 0cm 0cm,clip,width=3.1cm,height=3.1cm]{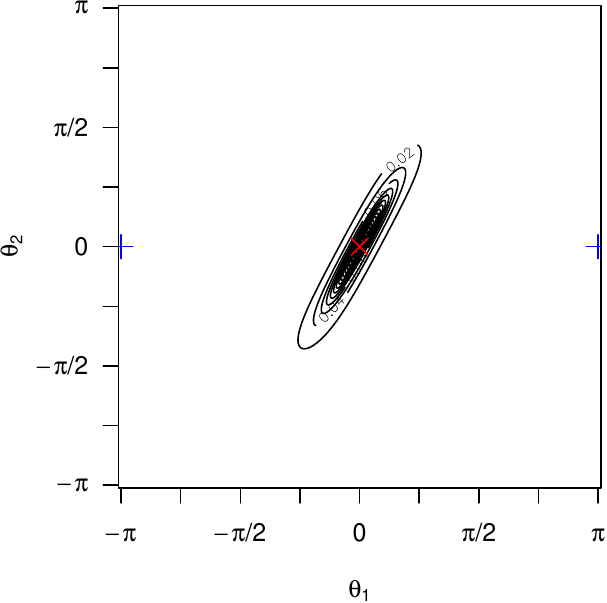}
		\includegraphics[trim=0.7cm 1cm 0cm 0cm,clip,width=3.1cm,height=3.1cm]{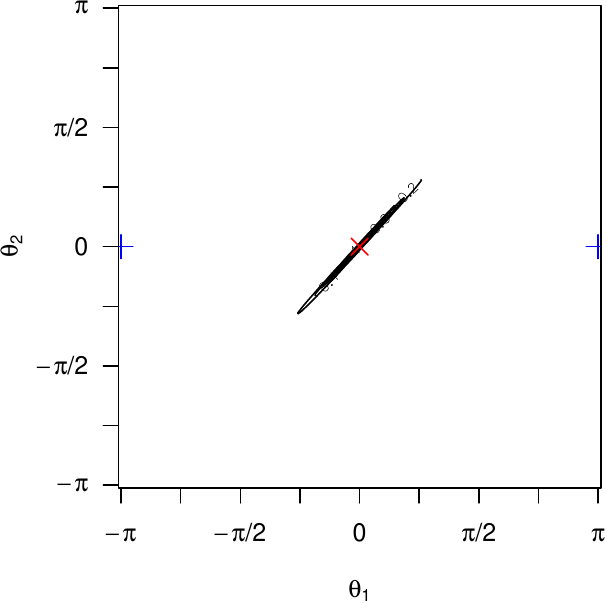} \vspace{0.2cm} \\ 
		} \caption[]{Contour plots 
		 of the density of {TWCC($\rho$)} given by \eqref{eq:tri_density} with $u_3=0$ and $\rho_{12}\rho_{13}\rho_{23}=1$ for nine combinations of $(\rho_{12}^*,\rho_{23.13}^*)$, each taking three different values.
		The horizontal axis represents the value of $u_1$, while the vertical axis stands for the value of $u_2$.
  The symbols `$\times$' (red) and `$+$' (blue) denote the modes and antimodes of density (\ref{eq:tri_density}), respectively.
		\label{fig:densities_new}}
\end{figure}

Figure \ref{fig:densities_new} plots the contour plots of the density (\ref{eq:tri_density}) with $u_3=0$ for various combinations of $(\rho_{12}^*,\rho_{23.13}^*)$.
Note that the values of the parameters $\rho_{12}^* = -5$ and $\rho_{23.13}^* = 0.1$ in one of the contour plots nearly correspond to those in Figure \ref{fig:densities}(a) of the main paper.
The strength of dependence between $u_1$ and $u_2$ increases with $|\rho_{23.13}^*|$.
Also, as $|\rho_{12}^*|$ increases, the shape of the density becomes closer to the limiting density (\ref{eq:limit}) which has  linear contours.
It appears that the convergence to the limiting density (\ref{eq:limit}) as $|\rho_{12}^*| \rightarrow \infty$ is faster for greater values of $|\rho_{23.13}^*|$.
If $|\rho_{12}^*|$ is close to 1, then the {TWCC} density appears to be more concentrated around $u_1=0$.

\section{Correlation coefficients} \label{sec:correlation_coef}
	We consider three well-known correlation coefficients for bivariate circular data.
	Let $(U_i, U_j)'$ be a bivariate circular random vector {and $X_{k} = (\cos U_{k}, \sin U_{k})' \ (k=i,j) $}.
	Then the correlation coefficients of Johnson and Wehrly \cite{suppJW77}, Jupp and Mardia \cite{suppJM80} and Fisher and Lee \cite{suppFL83} are respectively defined by 
	$$
	\rho_{{\rm JW}} = \lambda^{1/2}, \quad
	\rho_{{\rm JM}} =
	\mbox{tr} ( \Sigma_{ii}^{-1}
	\Sigma_{ij} \Sigma_{jj}^{-1}
	\Sigma_{ij}^T),
	$$
	$$ 
	\rho_{{\rm FL}} = \frac{\mbox{det} \{ E ( X_i X_j') \}}{[\mbox{det} \{ E ( X_i X_i')
		\} \mbox{det} \{ E( X_j X_j') \} ]^{1/2}},
	$$
	where $\lambda$ is
	the largest eigenvalue of $\Sigma_{ii}^{-1} \Sigma_{ij}
	\Sigma_{jj}^{-1} \Sigma_{ij}'$ and $\Sigma_{k \ell} = E ( X_{k}
	X_\ell') - E( X_{k}) E( X_\ell)'$ $(k,\ell=i,j)$.
	The correlation coefficients of our bivariate marginal distributions are given in the following theorem.
	
	\begin{thm} \label{thm:correlation}
		Let a trivariate random vector $(U_1,U_2,U_3)'$ follow the  model {TWCC($\boldsymbol{\rho}$)}.
		Then, for any pair of random variables $(U_i,U_j)'$, its correlation coefficients of Johnson and Wehrly \cite{suppJW77}, Jupp and Mardia \cite{suppJM80} and Fisher and Lee \cite{suppFL83} are given by
		$$
		{\rho_{{\rm JW}} = |\phi_{ij}|, \quad \rho_{{\rm JM}} = 2 \phi_{ij}^2 \quad \mbox{and} \quad \rho_{{\rm FL}} = \phi_{ij}^2, }
		$$ 
		respectively, where {$\phi_{ij}$ is given in  \eqref{thephi}.}
	\end{thm}

\begin{proof}
	If follows from Theorem \ref{thm:marginals}(i) that $(U_i,U_j)'$ has the density (\ref{eq:marginal_density}).
	This density is equivalent to a special case of the distribution of \cite{suppKP15} with 
 circular uniform marginals.
	Then the three correlation coefficients $\rho_{{\rm JW}}$, $\rho_{\rm JM}$ and $\rho_{\rm FL}$ of our model can be immediately calculated from those of the distribution of \cite{suppKP15} given in Section 2.6 of their paper.
\end{proof}
 
	Note the simplicity of the expressions for all the three correlation coefficients. Further  these all depend on the parameter  $\phi_{ij}$, implying that  it is a  good measure of dependence of the model. {Recall in this context the agreeable fact that $\phi_{ij}$ remains  invariant under $\rho_{ij}$ being replaced by $c\rho_{ij}$ for some constant $c$, hence  does not depend on the identifiability condition on $\rho_{ij}$.}

\section{Method of moments estimation}\label{sec:mom}
In this section, we discuss an alternative estimation procedure based on the results of Section~\ref{Sup_trigmom}.
Method of moments estimators can be obtained by equating theoretical and empirical trigonometric moments
$$
E \left\{ e^{{\rm i} (p_1 U_{1} + p_2 U_{2} + p_3 U_{3} )} \right\} = \frac{1}{n} \sum_{m=1}^n e^{{\rm i} (p_1 u_{1m} + p_2 u_{2m} + p_3 u_{3m}) }
$$
for some selected values of $(p_1,p_2,p_3)' \in \mathbb{Z}^3$.
In order to estimate the parameters of the original distribution {TWCC($\rhob$)}, possible choices of $(p_1,p_2,p_3)'$ are $(p_i,p_j,p_{k})'=(1,-1,0)'$ with $i<j$.
In this case, {equation (\ref{eq:tm_remark})} 
implies that (following lengthy but simple calculations for the second equality)
$$
E \left\{ e^{{\rm i} (U_{i} - U_{j})} \right\} = -\rho_{j k} \left( \frac{\varphi_{j k}}{\rho_{ij}} + \frac{1}{\rho_{i k}} \right) = \varphi_{ij}. 
$$
It follows that the parameters can be estimated as the solution of the following equations:
$$
{\hat{\varphi}_{ij}} = \frac{1}{n} \sum_{m=1}^n e^{{\rm i} (u_{im}-u_{jm} )}, \quad 1 \leq i < j \leq 3,
$$
where {$\hat{\varphi}_{ij}$} is the estimate of {$\varphi_{ij}$} defined in  \eqref{theVARPHI}.
Since there is no closed expression for $\{  \hat{\rho}_{ij} \}$, numerical derivation of these estimates from  $\{\varphi_{ij}\}$ is required. This is convoluted as the $\{  {\rho}_{ij} \}$  do not directly depend on  $\{\varphi_{ij}\}$ but through the two equations {related to  $\phi_{ij}$, see} \eqref{thephi} and~\eqref{theVARPHI}. Further, the parameter constraints \eqref{conditions} and \eqref{identcond} have to be satisfied, which makes the method of moments approach a numerically challenging task and we will not pursue it {since we could compute the maximum likelihood estimator as described in the section below}.

\section{Fisher Information} \label{sec:fisher_information}
Related to the maximum likelihood inference of Section \ref{sec: parameter estimation}, we now provide the associated Fisher information matrix. 
For ease of presentation, denote $c_2 = \frac{1}{(2\pi)^3} c_4^{1/2}$.
Then the expected Fisher Information matrix of the density \eqref{eq:tri_density} is given by
\begin{equation*}
    I(\rho_{12}, \rho_{13}, \rho_{23}) = n
    \begin{pmatrix}
        I_{\rho_{12}\rho_{12}} & I_{\rho_{12}\rho_{13}} & I_{\rho_{12}\rho_{23}} \\
        I_{\rho_{12}\rho_{13}} & I_{\rho_{13}\rho_{13}} & I_{\rho_{13}\rho_{23}} \\
        I_{\rho_{12}\rho_{23}} & I_{\rho_{13}\rho_{23}} & I_{\rho_{23}\rho_{23}} 
    \end{pmatrix}
\end{equation*}
where, denoting $(i,j,k)$ a permutation of $(1,2,3)$,
\begin{align*}
    & I_{\rho_{ij}\rho_{ij}} = \int_{[0,2\pi)^3} - \frac{\partial^2 \log\left(t(\ub_m;\rhob)\right)}{\partial \rho_{ij}^2}t(\ub_m;\rhob) du_1 du_2 du_3, \\
    & I_{\rho_{ij}\rho_{ik}} = \int_{[0,2\pi)^3} - \frac{\partial^2 \log\left(t(\ub_m;\rhob)\right)}{\partial \rho_{ij}\partial \rho_{ik}} t(\ub_m;\rhob) du_1 du_2 du_3, \\
  \mbox{and}\\  
    & \frac{\partial^2 \log\left(t(\ub_m;\rhob)\right)}{\partial \rho_{ij}^2} = \frac{\frac{\partial^2 c_2}{\partial \rho_{ij}^2} c_2 - \left(\frac{\partial c_2}{\partial \rho_{ij}}\right)^2}{c_2^2} - \frac{\frac{\partial^2 c_1}{\partial \rho_{ij}^2} F - \left( \frac{\partial c_1}{\partial \rho_{ij}} + 2\cos(u_{im} - u_{jm})\right)^2}{F^2}, \\
    & \frac{\partial^2 \log\left(t(\ub_m;\rhob)\right)}{\partial \rho_{ij}\partial \rho_{ik}} \\
    & \hspace{1cm} = \frac{\frac{\partial^2 c_2}{\partial \rho_{ij}\partial \rho_{ik}} c_2 - \frac{\partial c_2}{\partial \rho_{ij}} \frac{\partial c_2}{\partial \rho_{ik}}}{c_2^2} \\
    & \hspace{1.5cm} - \frac{\frac{\partial^2 c_1}{\partial \rho_{ij} \partial \rho_{ik}} F - \left( \frac{\partial c_1}{\partial \rho_{ij}} + 2\cos(u_{im} - u_{jm})\right) \left( \frac{\partial c_1}{\partial \rho_{ik}} + 2\cos(u_{im} - u_{k m})\right)}{F^2}.
    \end{align*}
Let us now write out these expressions in detail:
\begin{align*}   
    & F = c_1 + 2 \{ \rho_{12} \cos (u_{1m} - u_{2m}) + \rho_{13} \cos (u_{1m} - u_{3m}) + \rho_{23} \cos (u_{2m} - u_{3m}) \}, \\
    & \frac{\partial^2 c_1}{\partial \rho_{ij}^2} = 2\frac{\rho_{ik}\rho_{jk}}{\rho_{ij}^3}, \\
    & \frac{\partial^2 c_1}{\partial \rho_{ij} \partial \rho_{ik}} = \frac{1}{\rho_{jk}} - \frac{\rho_{jk}}{\rho_{ik}^2} - \frac{\rho_{jk}}{\rho_{ij}^2}, \\
    & \frac{\partial^2 c_2}{\partial \rho_{ij}^2} = \frac{1}{2} \frac{1}{(2\pi)^3} \left( \frac{\partial^2 c_4}{\partial \rho_{ij}^2} c_4^{-1/2} - \frac{1}{2} \left( \frac{\partial c_4}{\partial \rho_{ij}} \right)^2 c_4^{-3/2} \right),\\
    & \frac{\partial c_4}{\partial \rho_{ij}} = 2 \rho_{ij} \left( \left( \frac{\rho_{ik}}{\rho_{jk}} \right)^2 + \left( \frac{\rho_{jk}}{\rho_{ik}} \right)^2 \right) -  2 \frac{(\rho_{jk}\rho_{ik})^2}{\rho_{ij}^3} - 4 \rho_{ij}, \\
    & \frac{\partial^2 c_4}{\partial \rho_{ij}^2} = 2 \left( \frac{\rho_{ik}}{\rho_{jk}} \right)^2 + 2 \left( \frac{\rho_{jk}}{\rho_{ik}} \right)^2 +  6 \frac{(\rho_{jk}\rho_{ik})^2}{\rho_{ij}^4} - 4, \\
    & \frac{\partial^2 c_2}{\partial \rho_{ij}\partial \rho_{ik}} = \frac{1}{2} \frac{1}{(2\pi)^3} \left( \frac{\partial^2 c_4}{\partial \rho_{ij}\partial \rho_{ik}} c_4^{-1/2} - \frac{1}{2} \frac{\partial c_4}{\partial \rho_{ij}} \frac{\partial c_4}{\partial \rho_{ik}} c_4^{-3/2} \right),
    \end{align*}
    \begin{align*}
    \hspace{-6cm} \frac{\partial^2 c_4}{\partial \rho_{ij}\partial \rho_{ik}} = 4 \frac{\rho_{ij}\rho_{ik}}{\rho_{jk}^2} - 4 \rho_{ij}\frac{\rho_{jk}^2}{\rho_{ik}^3} - 4 \rho_{ik}\frac{\rho_{jk}^2}{\rho_{ij}^3}.
\end{align*}

\section{Monte Carlo Simulations} \label{sec:simulations}
In order to confirm that our MLE algorithm for the trivariate wrapped Cauchy copula from Section \ref{sec: parameter estimation} retrieves the true values of the parameters as the sample size increases, we have conducted a Monte Carlo simulation study. To this end, data has been generated from the copula in \eqref{eq:tri_density} using the algorithm described in Theorem \ref{thm:random}.
The sample sizes considered were $50,$  $100,$ $150,$ $200,$ $250,$ $300,$ $350,$ $500,$ $750,$ $1000,$ $1250,$ $1500,$ $1750,$ $2000,$ $3000,$ $4000,$ $5000$ and for each sample size we made 5000 replications.
The median of the results for the different lengths is presented in Figure~\ref{fig:simulations_increasing_n}. 
The true values of the parameters of the copula are $\rho_{12} = 1, \rho_{13} = 0.25$ and $\rho_{23} = 4$ and they are each plotted with dotted lines.
For each different sample, the MLE algorithm was repeated with 50 different initial values.
{The two smaller values converge faster to the truth compared to the larger one. However, the median of the values obtained for all three is close to the true value starting from sample size 500. }

\begin{figure}[htp]
    \centering
    \includegraphics[scale = 0.8]{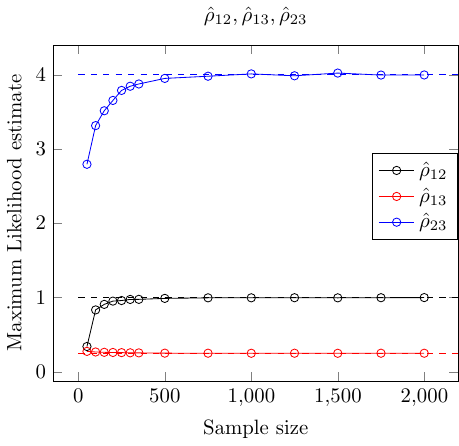}
    \caption{Plots of median values of maximum likelihood estimates for each of the parameters $\rho_{12}, \rho_{13}, \rho_{23}$ from 5000 replications for each sample size. The true parameters, which are plotted with horizontal dotted lines, are $\rho_{12} = 1, \rho_{13} = 0.25$ and $\rho_{23} = 4$.}
    \label{fig:simulations_increasing_n}
\end{figure}

Confidence intervals (CIs) for the estimates can be obtained by means of bootstrap.
Three different simulation studies were performed, with $B = 200$ bootstrap samples being obtained from a sample of size $n=100, n=500$ and $n=1000$, generated by the distribution with density \eqref{eq:tri_density}, for parameter values $\rho_{12} = 1, \rho_{13} = -4$ and $\rho_{23} = -0.25$.
The estimates of the parameters are obtained by repeating the procedure described in Section \ref{sec: parameter estimation} 200 times.
The median of the values and the 95\% CI for the estimated parameter values are shown in Table \ref{tab:bootstrap}.
The true values of the parameter are included in all bootstrap CIs. 
For all values of the sample size, the true value is included in the bootstrap CI.
{Even though the value 0 is included in some CIs, $\rho_{ij}$ cannot take this value and so we write $(a,b)$ for $(a,b) \setminus \{0\}$ where $a<0<b$.}

\begin{table}[htb]
    \centering
    \caption{The true value of the parameters is reported along with the median and 95\% bootstrap confidence interval of the ML estimates of the parameters for 200 bootstrap samples and sample size $n=100, n=500$ and $n=1000$.}
\footnotesize{
    \begin{tabular}{|c|cccc|}
    \hline
        Parameter & True value & \multicolumn{3}{c|}{Median ($95\%$ bootstrap CI)} \\
        && $n=100$ & $n=500$ & $n=1000$ \\
        \hline
         $\rho_{12}$ & 1 & 1.83 (-0.27, 3.96) & 0.99 (-0.19, 1.32) & 1.06 (0.50, 1.31) \\
         $\rho_{13}$ & -4 & -1.92 (-5.26, 0.31) & -3.92 (-8.97, 2.71) & -3.74 (-6.26, -2.63) \\
         $\rho_{23}$ & -0.25 & -0.24 (-2.36, 2.33) & -0.26 (-1.12, 0.74) & -0.25 (-0.29, -0.18) \\
         \hline
    \end{tabular}
}
    \label{tab:bootstrap}
\end{table}

\section{A multivariate extension of our model}\label{sec: possible_extensions}

It is natural albeit highly challenging to extend the model presented here to any $d$-dimensional torus. A potential model could be of the form 
\begin{equation}\label{multi}
t(\ub; \rhob)\propto\left\{c_4+2\sum_{1\leq i<j\leq d}\rho_{ij}\cos(u_i-u_j)\right\}^{-1}
\end{equation}
where $\ub = (u_1, \ldots, u_d)'$, $\rhob \in \R^{d(d-1)/2}$ contains the parameters $\rho_{ij}\in\mathbb{R}$, which need to satisfy certain conditions and $c_4$ depends on them. To make this a valid density, $c_4$ has at least to be equal to $2\sum_{1\leq i<j\leq d}|\rho_{ij}|$, and we need to find the normalizing constant. 
Note that the general form of~\eqref{multi} has been proposed  in \cite{suppMS20}  Equation (106); no properties were investigated. One of the key properties of the  copula \eqref{multi} is given in the following theorem.
\begin{thm} \label{thm:multivariate}
	Let a $[0,2\pi)^d$-valued random vector $(U_1,\ldots,U_d)'$ have the probability density function $t_d(\ub)$ for $\ub = (u_1,\ldots,u_d)'$.
	Suppose that $t_d$ is a function of $\{ u_i - u_j \, ; \, 1 \leq i < j \leq d\}$, namely,
	\begin{equation}
	t_d(\ub) = h (u_1-u_2 , u_1-u_3 , \ldots , u_{d-1}-u_d), \quad 0 \leq u_1,\ldots,u_d < 2\pi.  \label{eq:general_model}
	\end{equation}
{Then the following hold for $(U_1,\ldots,U_d)'$.

\begin{itemize}
\item[(i)] The marginal distribution of $U_i$ $(1 \leq i \leq d)$ has the uniform distribution on the circle.
\item[(ii)] If $p_1+\cdots+p_d \neq 0$, the trigonometric moments of order $(p_1,\ldots,p_d)'$ $(p_1,\ldots,p_d \in \mathbb{Z})$ are
$
E[e^{{\rm i} (p_1 U_1 + \cdots + p_d U_d)}] = 0.
$
\end{itemize}
}
\end{thm}

\begin{proof}
(i) Without loss of generality, we calculate the marginal distribution of $U_d$.
The marginal density of $U_d$ can be expressed as
\begin{align}
 f_{u_d} (u_d) = \int_{[0,2\pi)^{d-1}} h (u_1-u_2,u_1-u_3,\ldots,u_{d-1}-u_d) du_1 \cdots du_{d-1}. \label{eq:multi_proof}
\end{align}
Putting $v_i = u_i - u_d$ $(i=1,\ldots,d-1)$, it follows that $u_i - u_j = v_i - v_j $ for $j=1,\ldots,d-1$ and therefore $\{ u_i-u_j \, ; \, 1 \leq i < j \leq d \}$ can be expressed using $d-1$ variables $v_1,\ldots,v_{d-1}$.
Thus the marginal density (\ref{eq:multi_proof}) can be expressed as
$$
f_{u_d}(u_d) = \int_{[0,2\pi)^{d-1}} h (v_1-v_2 , v_1 -v_3, \ldots, v_{d-1} ) dv_1 \cdots dv_{d-1} = C.
$$
Since $f(u_d)$ is a constant which does not depend on $u_d$, it follows that the marginal distribution of $U_d$ is the uniform distribution on the circle.

{
(ii) Using $v_1,\ldots,v_{d-1}$ and $u_d$, the trigonometric moments of order $(p_1,\ldots,p_d)$ can be expressed as
\begin{align*}
\lefteqn{ E \left[ e^{ {\rm i} (p_1 U_1 + \cdots p_d U_d)} \right] } \hspace{1cm} \\
= & \int_{[0,2\pi)^d} e^{{\rm i} (p_1 u_1 + \cdots + p_d u_d)} h (u_1 , \ldots , u_d ) du_1 \cdots du_d \\
= &  \int_{[0,2\pi)^d} e^{{\rm i} \{ p_1 (v_1+u_d) + \cdots + p_{d-1} (v_{d-1}+u_d) + p_d u_d\} } \\
& \times h (v_1-v_2 , v_1 -v_3, \ldots, v_{d-1} ) dv_1 \cdots dv_{d-1} du_d
 \\
= & \int_{[0,2\pi)^{d-1}} e^{{\rm i} ( p_1 v_1 + \cdots + p_{d-1} v_{d-1}) } h (v_1-v_2 , v_1 -v_3, \ldots, v_{d-1} ) dv_1 \cdots dv_{d-1} \\
& \times \int_{[0,2\pi)} e^{{\rm i} \{ (p_1+ \cdots + p_d) u_d \} } du_d.
\end{align*}
Since $p_1+\cdots + p_d \neq 0$, we have $\int_{[0,2\pi)} e^{{\rm i} \{ (p_1+ \cdots + p_d) u_d \} } du_d=0.$
Therefore 
$$
E \left[ e^{ {\rm i} (p_1 U_1 + \cdots p_d U_d)} \right] = 0.
$$
}
\end{proof}

{Note that Theorem \ref{thm:multivariate}(ii) is an extension of Theorem \ref{thm:moments}(i) which provides the trigonometric moments of the {TWCC} for $p_1+p_2+p_3 \neq 0$.}
These general results  make it appealing to work on the extension of {TWCC} to~\eqref{multi} in the future.

\section{A generalization of the TWC copula} \label{sec:generalc1}

Here we consider a generalization of the {TWC copula (\ref{eq:tri_density}) to allow for greater flexibility.}
This generalization can be derived by removing the constraint on $c_1$ of the density (\ref{eq:tri_density}) and is given by
\begin{align}
	\lefteqn{ c(u_1,u_2,u_3) } \nonumber \\
    & = C_2 \Bigl[ C_1 + 2 \left\{ \rho_{12} \cos (u_1 - u_2) + \rho_{13} \cos (u_1 - u_3) + \rho_{23} \cos (u_2 - u_3) \right\} \Bigr]^{-1}, \nonumber \\
	& \hspace{8cm} 0 \leq u_1,u_2,u_3 < 2\pi, \label{eq:tri_density_g} 
\end{align}
where $\rho_{12},\rho_{13},\rho_{23} \in \mathbb{R}$ satisfy $|\rho_{j k}| < |\rho_{ij} \rho_{i k}| / ( |\rho_{ij}| + |\rho_{i k}|)$ for some $(i,j,k)$, that is, a permutation of $(1,2,3)$, $C_1>2(|\rho_{ij}| + |\rho_{i k}|-|\rho_{j k}|)$, $C_2 = (4 \pi)^{-2} \alpha_2$ $ \times [ K \{ (\alpha_1+\alpha_2^2/2)^{1/2} / \alpha_2 \} ]^{-1}$, $\alpha_1 = - C_1^2/2 + 2 \rho_{12}^2 + 2 \rho_{13}^2 + 2 \rho_{23}^2$,
{ \begin{align*}
	\alpha_2 = & \ 2 \, \Biggl\{ \left( \frac{C_1}{2} + \rho_{12} + \rho_{13} + \rho_{23} \right) \times \prod_{(i,j,k) \in T}   \left( \frac{C_1}{2} + \rho_{ij} - \rho_{i k} - \rho_{j k} \right) \Biggr\}^{1/4},
\end{align*} 
$T=\{(1,2,3),(1,3,2),(2,3,1)\}$}, and $K$ denotes the complete elliptic integral of the first kind (e.g., \cite[equation (8.112.1)]{suppGR2007}) given by
$$
K(\alpha) = \int_0^{\pi/2} \frac{1}{(1-\alpha^2 \sin^2 t)^{1/2}} d t.
$$

{If $C_1=c_1$, then the argument of $K$ in $C_2$ equals 0 and therefore $C_2=c_2$.
This can be proved as follows.
First it can be shown that, with $C_1=c_1$,
$$
\frac{c_1}{2} + \rho_{ij} + q \rho_{i k} +q \rho_{j k} = 
\frac{\left( \rho_{i k} \rho_{j k} + q \rho_{ij} \rho_{i k} + q \rho_{i j} \rho_{j k} \right)^2}{2 \rho_{12} \rho_{13} \rho_{23}} , \quad q =-1,1.
$$
Hence
\begin{align*}
\alpha_2^2 & = 4  \, \Biggl\{ \left( \frac{c_1}{2} + \rho_{12} + \rho_{13} + \rho_{23} \right) \times \prod_{(i,j,k) \in T}   \left( \frac{c_1}{2} + \rho_{ij} - \rho_{i k} - \rho_{j k} \right) \Biggr\}^{1/2} \\
& = \frac{1}{\rho_{12}^2 \rho_{13}^2 \rho_{23}^2} \left\{ (\rho_{12}\rho_{13} + \rho_{12}\rho_{23} + \rho_{13} \rho_{23}) \times \prod_{(i,j,k) \in T} (\rho_{i k} \rho_{j k} - \rho_{ij} \rho_{i k} - \rho_{i j} \rho_{j k}) \right\},
\end{align*}
which can be shown to be equal to
$$
 \alpha_2^2 = c_1^2 - 4 \rho_{12}^2 - 4 \rho_{13}^2 - 4 \rho_{23}^2,
$$
so 
$ \alpha_1 + \alpha_2^2 /2 = 0$, and the argument of $K$ in $C_2$ is equal to 0.
Hence the normalizing constant reduces to
$$
C_2 = \frac{ \alpha_2 }{(4\pi)^2 K(0)} = \frac{\{ c_1^2 - 4 \rho_{12}^2 - 4 \rho_{13}^2 - 4 \rho_{23}^2 \}^{1/2}}{(4\pi)^2 \cdot \pi/2 } = c_2.
$$
}

It is possible to generalize the model (\ref{eq:tri_density}) further by lifting the condition on the parameters $|\rho_{j k}| < |\rho_{ij} \rho_{i k}| / ( |\rho_{ij}| + |\rho_{i k}|)$.
However the domain of $C_1$ is more involved in that case.

The generalized distribution (\ref{eq:tri_density_g}) shares some tractable properties of the distribution (\ref{eq:tri_density}).
For example, the following hold for marginal and conditional distributions of the distribution (\ref{eq:tri_density_g}).
\begin{thm}
	Let $(U_1,U_2,U_3)'$ follow the distribution (\ref{eq:tri_density_g}).
	Then the following hold for the marginal or conditional distributions of $(U_1,U_2,U_3)'$.
	
	\begin{enumerate}
		\item[(i)] The marginal distribution of $U_i$ is the uniform distribution on the circle.
		\item[(ii)] The conditional distribution of $(U_i,U_j)'$ given $U_{k}=u_{k}$ is the special case of the distribution of \cite{suppKP15} with density (\ref{eq:conditional_1st}).
		\item[(iii)] The conditional distribution of $U_i$ given $(U_j,U_{k})'=(u_j,u_{k})'$ is the wrapped Cauchy distribution (\ref{eq:conditional_density2}).
	\end{enumerate}
\end{thm}

\begin{proof}
	 The property (i) is clear from Theorem \ref{thm:multivariate}.
	The properties (ii) and (iii) follow immediately from Theorem \ref{thm:conditionals}(i) and (iii), respectively, because the generalized density (\ref{eq:tri_density_g}) has the same functional form as the density (\ref{eq:tri_density}). 
\end{proof}

It is remarked that differences between the generalized model (\ref{eq:tri_density_g}) and the original one (\ref{eq:tri_density}) are that the marginal distribution of $(U_i,U_j)'$ is not a submodel of \cite{suppWJ80} and that the conditional distribution of $U_i$ given $U_j=u_j$ is not the wrapped Cauchy in general.

The expressions of modes and antimodes for the generalized density (\ref{eq:tri_density_g}) are as simple as those for the original density (\ref{eq:tri_density}).
The proof is straightforward from Theorem (\ref{thm:modes}) and is therefore omitted.
\begin{cor} \label{cor:g_modes}
	The modes and antimodes of the density (\ref{eq:tri_density_g}) are the same as those of the density (\ref{eq:tri_density}) given in Theorem \ref{thm:modes}.
\end{cor}
Note that this corollary leads to the simplicity of the range of the parameter $C_1>2(|\rho_{ij}| + |\rho_{i k}|-|\rho_{j k}|)$.


\end{document}